\documentclass[11pt]{article}
\usepackage{setspace}
\usepackage{amsfonts,amsmath,bm,xspace,amsthm,fullpage,amssymb}
\usepackage{multirow,graphicx, algorithm,algorithmic,rotating}
\usepackage{url,cite}
\usepackage{enumerate}

\usepackage{color,soul}

\usepackage{hyperref}
\usepackage[small,bf]{caption}

\newtheorem{theorem}{Theorem}
\newtheorem{lemma}{Lemma}

\newtheorem{proposition}{Proposition}

\newtheorem{definition}{Definition}

\newtheorem*{remark}{Remark}

\numberwithin{equation}{section}

\DeclareMathOperator{\conv}{conv}

\newcommand{\A}{\mathcal{A}} 
\newcommand{\C}{\mathbb{C}}
\newcommand{\R}{\mathbb{R}}
\newcommand{\vnorm}[1]{\lVert#1\rVert}

\newcommand{\minimize}{\mathop{\operatorname{minimize}}}
\newcommand{\maximize}{\mathop{\operatorname{maximize}}}

\newcommand\numberthis{\addtocounter{equation}{1}\tag{\theequation}}

\title{\huge Near Minimax Line Spectral Estimation}

\author{Gongguo Tang$^\dagger,$ Badri Narayan Bhaskar$^\dagger,$
and Benjamin Recht$^\sharp$\\
$^\dagger$Department of Electrical and Computer Engineering\\
$^\sharp$Department of Computer Sciences\\
University of Wisconsin-Madison
}

\date{February 2013; Last Revised March 2013.}

\begin{document}

\maketitle

\vspace{-0.3in}

\begin{abstract} This paper establishes a nearly optimal algorithm for
estimating the frequencies and amplitudes of a mixture of sinusoids from noisy
equispaced samples. We derive our algorithm by viewing line spectral estimation
as a sparse recovery problem with a continuous, infinite dictionary. We show
how to compute the estimator via semidefinite programming and provide
guarantees on its mean-square error rate. We derive a complementary minimax
lower bound on this estimation rate, demonstrating that our approach nearly
achieves the best possible estimation error. Furthermore, we establish bounds
on how well our estimator localizes the frequencies in the signal, showing that
the localization error tends to zero as the number of samples grows.  We verify our theoretical results in an array of numerical experiments, demonstrating that the semidefinite programming approach outperforms two classical spectral estimation techniques.
\end{abstract}

\section{Introduction}

Spectrum estimation is one of the fundamental problems in statistical signal
processing. Despite of hundreds of years of research on this subject, there
still remain several fundamental open questions in this area. This paper
addresses a central one of these problems: how well can we determine the
locations and magnitudes of spectral lines from noisy temporal samples? In this
paper, we establish lower bounds on how well we can recover such signals and
demonstrate that these worst case bounds can be nearly saturated by solving a
convex programming problem. Moreover, we prove that the estimator approximately
localizes the frequencies of the true spectral lines.

We consider signals whose spectra consist of spike trains with unknown
locations in a normalized interval $\mathbb{T} = [0, 1]$. Consider $n=2m+1$
equispaced samples of a mixture of sinusoids given by
\begin{equation}
\label{signal}
x^\star_j =  \sum_{l=1}^k c_l \exp(i 2\pi j f_l)
\end{equation}
where $j \in \{-m,\dots, m\}$. We assume that the support $T = \{f_l\}_{l=1}^k
\subset \mathbb{T}$ of the $k$ frequencies and the corresponding complex
amplitudes $\{c_l\}_{l=1}^k$ are unknown. We observe noisy samples $y = x^\star
+ w$ where the noise components $w_i$ are i.i.d. centrally symmetric complex
Gaussian variables with variance $\sigma^2$. By swapping the roles of frequency
and time or space, the signal model \eqref{signal} also serves as a proper
model for superresolution imaging where we aim to localize temporal events or
spatial targets from noisy, low-frequency
measurements~\cite{cg_exact12,cg_noisy}. Our first result characterizes the
denoising error $\frac{1}{n}\vnorm{x^\star - \hat{x}}_2^2$ and is summarized in
the following theorem.

\begin{theorem}
\label{main}
Suppose the line spectral signal $x^\star$ is given by \eqref{signal} and we
observe $n$ noisy consecutive samples $y_j = x^\star_j + w_j$ where $w_j$ is
i.i.d. complex Gaussian with variance $\sigma^2$. If the frequencies
$\{f_l\}_{l=1}^k$ in $x^\star$ satisfy a minimum separation condition
\begin{equation}
\label{min-sep}
\min_{p\neq q}d(f_p,f_q) > 4/n
\end{equation}
with $d(\cdot, \cdot)$ the distance metric on the torus, then we can determine
an estimator $\hat{x}$ satisfying
\begin{equation}
\label{fast-rate}
\frac{1}{n}\vnorm{\hat{x} - x^\star}_2^2 = O\left( \sigma^2 \frac{k \log(n)}{n}\right)  
\end{equation}
with high probability by solving a semidefinite programming problem.
\end{theorem}

Note that if we exactly knew the frequencies $f_j$, the best rate of estimation
we could achieve would be $O(\sigma^2 k / n)$~\cite{oracle_lasso}. Our upper
bound is merely a logarithmic factor larger than this rate. On the other hand,
we will demonstrate via minimax theory that a logarithmic factor is unavoidable
when the support is unknown. Hence, our estimator is nearly minimax optimal.

It is instructive to compare our stability rate to the optimal rate achievable
for estimating a sparse signal from a finite, discrete
dictionary~\cite{cd_minimax}. In the case that there are $p$ incoherent
dictionary elements, no method can estimate a $k$-sparse signal from $n$
measurements corrupted by Gaussian noise at a rate less than $O(\sigma^2
\frac{k\log(p/k)}{n})$. In our problem, there are an infinite number of
candidate dictionary elements and it is surprising that we can still achieve
such a fast rate of convergence with our highly coherent dictionary. We
emphasize that none of the standard techniques from sparse approximation can be
immediately generalized to our case. Not only is our dictionary infinite, but
also it does not satisfy the usual assumptions such as restricted eigenvalue
conditions~\cite{rest_eig} or coherence conditions~\cite{coherence} that are
used to derive stability results in sparse approximation. Nonetheless, in terms
of mean-square error performance, our results match those obtained when the
frequencies are restricted to lie on a discrete grid.

In the absence of noise, polynomial interpolation can exactly recover a line
spectral signal of $k$ \emph{arbitrary} frequencies with as few as $2 k$
equispaced measurements. In the light of our minimum frequency separation
requirement~\eqref{min-sep}, why should one favor convex techniques for line
spectral estimation? Our stability result coupled with minimax optimality
establish that no method can perform better than convex methods when the
frequencies are well-separated. And, while polynomial interpolation and
subspace methods do not impose any resolution limiting assumptions on the
constituent frequencies, these methods are empirically highly sensitive to
noise. To the best of our knowledge, there is no result similar to
Theorem~\ref{main} that provides finite sample guarantees about the noise
robustness of polynomial interpolation techniques.

Additionally, little is known about how well spectral lines can be localized
from noisy observations. The frequencies estimated by any method will never
exactly coincide with the true frequencies in the signal in the presence of
noise. However, we can characterize the localization performance of our convex
programming approach, and summarize this performance in Theorem~\ref{support}.

Before stating the theorem, we introduce a bit of notation. Define
neighborhoods $N_j$ around each frequency $f_j$ in $x^\star$ by $N_j := \{ f
\in \mathbb{T} : d(f,f_j) \leq 0.16/n\}$. Also define $F = \mathbb{T}\backslash
\cup_{j=1}^k N_j$ as the set of frequencies in $\mathbb{T}$ which are not near
any true frequency. The letters $N$ and $F$ denote the regions that are
\emph{near} to and \emph{far} from the true supporting frequencies. The
following theorem summarizes our localization guarantees.

\begin{theorem}
\label{support}
Let $\hat{x}$ be the solution to the same semidefinite programming (SDP)
problem as referenced in Theorem~\ref{main} and $n > 256$. Let $\hat{c_l}$ and
$\hat{f}_l$ form the decomposition of $\hat{x}$ into coefficients and
frequencies, as revealed by the SDP. Then, there exist fixed numerical
constants $C_1,C_2$ and $C_3$ such that with high probability
\begin{enumerate}[i.)]
\item $\sum_{l : \hat{f}_l \in F} |\hat{c}_l| \leq  C_1 \sigma\sqrt{\frac{k^2 \log(n)}{n}}$
\item $\sum_{l : \hat{f}_l \in N_j} |\hat{c}_l| \left\{ \min_{f_j \in T} d(f_j,\hat{f}_l) \right\}^2 \leq  C_2 \sigma\sqrt{\frac{k^2 \log(n)}{n}}$
\item $\left| c_j - \sum_{l : \hat{f_l} \in N_j} \hat{c}_l \right| \leq C_3 \sigma\sqrt{\frac{k^2 \log(n)}{n}}$.
\item If for any frequency $f_j$, the corresponding amplitude $|c_j| > C_1
\sigma \sqrt{\frac{ k^2 \log(n)}{n}}$, then with high probability there exists
a corresponding frequency $\hat{f}_j$ in the recovered signal such that,
\[
\left| f_j - \hat{f}_j \right|  \leq \frac{\sqrt{C_2/C_1}}{n}\left(\frac{|c_j|}{C_1 \sigma \sqrt{\frac{ k^2 \log(n)}{n}}} - 1\right)^{-\tfrac{1}{2}}
\]
\end{enumerate}
\end{theorem}
  
Part (i) of Theorem~\ref{support} shows that the estimated amplitudes
corresponding to frequencies far from the support are small. In practice, we
note that we rarely find any spurious frequencies in the far region, suggesting
that our bound (i) is conservative. Parts (ii) and (iii) of the theorem show
that in a neighborhood of each true frequency, the recovered signal has
amplitude close to the true signal. Part (iv) shows that the larger a
particular coefficient is, the better our method is able to estimate the
corresponding frequency. In particular, note that if $|c_j| > 2 C_1 \sigma
\sqrt{\frac{ k^2 \log(n)}{n}}$, then $\left| f_j - \hat{f}_j \right| \leq
\frac{\sqrt{C_2/C_1} }{n}$. In all four parts, note that the localization error
goes to zero as the number of samples grows.

We proceed as follows. In Section~\ref{sec:line-spec}, we begin by
contextualizing our result in the canon of line spectral estimation. We
emphasize the advantages and shortcomings of prior art, and describe the
methods on which our analysis is built upon. We then in
Section~\ref{sec:atomic-norms} describe the semidefinite programming approach
to line spectral estimation, originally introduced in~\cite{btr12}, and explain
how it relates to other recent spectrum estimation algorithms. We present
minimax lower-bounds for line spectral estimation in Section~\ref{sec:minimax}.
We then provide the proofs of our main results in Section~\ref{sec:proofs}.
Finally, in Section~\ref{sec:experiments}, we empirically demonstrate that the
semidefinite programming approach outperforms MUSIC~\cite{music} and Cadzow's
technique~\cite{cadzow05} in terms of the localization metrics defined by parts
(i), (ii) and (iii) of Theorem~\ref{support}.

\section{Prior Art in Line Spectral Estimation}
\label{sec:line-spec}

To date, line spectral analysis may be broadly classified into two camps.
\emph{Subspace methods}~\cite{music,esprit,cadzow05,ssa} build upon polynomial
interpolation~\cite{prony1795} and exploit certain low rank structure in the
spectrum estimation problem for denoising. Research on subspace approaches has
yielded several standard algorithms that are widely deployed and shown to
achieve Cram\'{e}r-Rao bound asymptotically~\cite{fri,cramer-subspace}.
However, the sensitivity to noise and model order is not well understood, and
there are few guarantees of how these algorithms perform given a limited number
of noisy measurements. For a review of many of these classical approaches, see
for example~\cite{StoicaMoses}.

More recently, approaches based on convex optimization have gained favor and
have been demonstrated to perform well on a variety of spectrum estimation
tasks~\cite{malioutov05,bourguignon2007irregular,baraniuk2010model,zweig2003irre
gular}. These convex programming methods restrict the frequencies to lie on a
finite grid of points and view line spectral signals as a sparse combination of
single frequencies. While these methods are reported to have significantly
better localization properties than subspace methods (see for
example,~\cite{malioutov05}) and admit fast and robust algorithms, they have
two significant drawbacks. First, while finer gridding may lead to better
performance, very fine grids are often numerically unstable. Furthermore,
traditional compressed sensing theory does not adequately characterize the
performance of fine gridding in these algorithms as the dictionary becomes
highly coherent.

Some very recent work~\cite{btr12,cg_exact12,cg_noisy} bridges the gap between
the performant discretized algorithms and continuous subspace approaches by
developing a new theory of convex relaxations for infinite continuous
dictionary of frequencies. Our work in~\cite{btr12} applies the atomic norm
framework proposed by Chandrasekaran et al~\cite{crpw} to the line spectral
estimation problem. There, we established stability results on the denoising
error and demonstrated empirically that our algorithm compared favorably with
both the classical and recent convex approaches which assume the frequencies
are on an oversampled DFT grid. Our prior results made no assumption about the
separation between frequencies. When the frequencies are well separated, the
current work demonstrates that much faster convergence rates are achieved.

Our work is closely related to recent results established by Cand\`es and
Fernandez-Granda~\cite{cg_exact12} on exact recovery using convex methods and
their recent work~\cite{cg_noisy} on exploiting the robustness of their dual
polynomial construction to show super-resolution properties of convex methods.
The total variation norm formulation used in~\cite{cg_noisy} is equivalent to
the atomic norm specialized to the line spectral estimation problem.

Robustness bounds were established in both our earlier work~\cite{btr12} and in
the work of Cand\`es and Fernandez-Granda~\cite{cg_noisy}. In~\cite{btr12}, a
slow convergence rate was established with no assumptions about the separation
of frequencies in the true signal. In~\cite{cg_noisy}, the authors provide
guarantees on the $L_1$ energy of error in the frequency domain in the case
that the frequencies are well separated. The noise is assumed to be adversarial
with a small $L_1$ spectral energy. In contrast, our paper shows near minimax
denoising error under Gaussian noise. It is also not clear that there is a
computable formulation for the optimization problem analyzed
in~\cite{cg_noisy}. While the guarantees the authors derive in~\cite{cg_noisy}
are not comparable with our results, several of their mathematical
constructions are used in our proofs here.

Additional recent work derives conditions for approximate support recovery
under the Gaussian noise model using the Beurling-Lasso ~\cite{azais}. There,
the authors show that there is a true frequency in the neighborhood of every
estimated frequency with large enough amplitude. We note that the
Beurling-Lasso is equivalent to the atomic norm algorithm that we analyze in
this paper. A more recent paper by Fernandez-Granda{\cite{granda2}} improves
this result by giving conditions on recoverability in terms of the true signal
instead of the estimated signal and prove a theorem similar to
Theorem~\ref{support}, but use a worst case $L_2$ bound on the noise samples.
Here, we improve these recent results in our proof of Theorem~\ref{support},
providing tighter guarantees under the Gaussian noise model.
 
\section{Frequency Localization using Atomic Norms}
\label{sec:atomic-norms}

We describe more precisely our signal model in this section. Suppose we wish to
estimate the amplitudes and frequencies of a signal $x(t), t \in \R$ given as a
mixture of $k$ complex sinusoids:
\begin{equation*}
  x ( t) =  \sum_{l = 1}^k c_l \exp ( i 2 \pi f_l t)
\end{equation*}
where $\{ c_l \}_{l = 1}^k$ are unknown complex amplitudes corresponding to
the $k$ unknown frequencies $\{ f_l \}_{l = 1}^k$ assumed to be in the torus
$\mathbb{T} = [0, 1]$. Such a signal may be thought of as a normalized band
limited signal and has a Fourier transform given by a line spectrum:
\begin{equation}
\label{mu}
\mu(f) = \sum_{l=1}^k c_l\delta(f - f_l)
\end{equation}
Denote by $x^\star$ the $n = 2m+1$ dimensional vector composed of equispaced 
Nyquist samples $\{x(j)\}_{j=-m}^m$   for $j=-m,\ldots,m$.

The goal of line spectral estimation is to estimate the frequencies and 
amplitudes of the signal $x(t)$ from the finite, noisy samples $y \in \C^n$ 
given by
\[
  y_j  =  x_j^\star + w_j\\
\]
for $-m \leq j \leq m$, where $w_j \sim \mathcal{C}\mathcal{N}(0,\sigma^2)$ is 
i.i.d. circularly symmetric complex Gaussian noise. 

\subsection{Algorithm: Atomic Norm Soft Thresholding (AST)}\label{sec:algorithms}

We can model the line spectral observations $x^\star = [x_{-m}^\star,
\ldots,x_{m}^\star]^T \in \C^n$ as a sparse combination of ``atoms'' $a(f)$ 
which correspond to observations due to single frequencies.  Define the vector 
$a(f) \in \C^n$ for any $f \in \mathbb{T} = [0,1]$ by
\[
a(f) = \begin{bmatrix}
e^{i2\pi (-m)f}\\
\vdots\\
1\\
\vdots\\
e^{i2\pi m f}
\end{bmatrix} \in \C^n.
\]
Then, we rewrite model \eqref{signal} as follows:
\begin{align}
\label{model}x^\star = \sum_{l=1}^k c_l a(f_l) = \sum_{l=1}^k |c_l| 
a(f_l)e^{i\phi_l}
\end{align}
where $\phi_l = c_l/|c_l|$ is the phase of the $l$th component. So, the target 
signal $x^\star$ may be viewed as a sparse non-negative combination of elements 
from the atomic set $\A$ given by
\begin{align}
\A = \left\{ a(f)e^{i \phi}, f \in [0,1], \phi \in [0,2\pi] 
\right\}.\label{atomicset}
\end{align}
For a general atomic set $\A$, the atomic norm of a vector is defined as the 
gauge function associated with the convex hull $\conv(\A)$ of atoms:
\begin{align}
 \vnorm{z}_\A = \inf \left\{ t > 0: z \in t \conv(\A)\right\}
\label{def-atnorm} = \inf \left\{\sum_a {c_a} : z = \sum_a {c_a a}, a \in \A, 
c_a > 0 \right\}
\end{align}
The authors in \cite{crpw} justify the use of atomic norm $\vnorm{\cdot}_\A$ as
a general penalty function to promote sparsity in an infinite dictionary $\A$.
This generalizes various forms of sparsity. For example, the $\ell^1$
norm~\cite{CRT06} for sparse vectors is an atomic norm corresponding to the
atomic set formed by canonical unit vectors. The nuclear
norm~\cite{CandesRecht09} for low rank matrices is an atomic norm induced by
the atomic set of unit-norm rank-1 matrices.

In this paper, we analyze the performance of the atomic norm soft thresholding 
(AST)  estimate:
\begin{equation}
\label{atnorm-denoise}
\hat{x} = \arg\min_z \frac{1}{2} \vnorm{y-z}_2^2 + \tau \vnorm{z}_\A
\end{equation}
where the atomic norm $\|\cdot\|_\A$ corresponds to the atomic set in 
\eqref{atomicset}, and $\tau$ is a suitably chosen regularization parameter. 
The corresponding dual problem is interesting because it gives a way of 
localizing the frequencies in an atomic norm achieving decomposition of 
$\hat{x}$. The dual problem of AST is given by the following semi-infinite 
program:
\begin{align}
\label{atnorm-dual}
\maximize_q~ &\frac{1}{2}\vnorm{y}_2^2 - \frac{1}{2}\vnorm{y - \tau 
q}_2^2\nonumber\\
\text{subject to } & \sup_{f \in \mathbb{T}}~ \left|\langle q, a(f) 
\rangle\right| \leq 1
\end{align}
It is convenient to associate a trigonometric polynomial $\hat{Q}(f) = \langle 
\hat{q}, a(f) \rangle$ with the optimal solution $\hat{q}$ of the dual problem. 
As discussed in~\cite{btr12}, the frequencies in the support of the solution 
$\hat{x}$ can be identified by finding points on the torus $\mathbb{T}$ where 
$\hat{Q}$ has a magnitude of unity. We use 
\begin{equation}
\hat{x} = \sum_{l} \hat{c}_l a(\hat{f}_l)
\end{equation}
to denote the decomposition of $\hat{x}$ given by the dual polynomial 
$\hat{Q}(f)$.

We show in~\cite{btr12} that a good choice of $\tau$ for obtaining accelerated 
convergence rates is 
\begin{equation}
\label{tau}
\tau = \eta \sigma\sqrt{n \log(n)}
\end{equation}
for some $\eta \in (1, \infty)$. We shall use this choice of regularization
parameter throughout this paper.
\begin{remark}
As shown in Section III.A of our prior work\cite{btr12},
problem~\eqref{atnorm-denoise} is equivalent to the semidefinite programming
problem
\begin{align}\label{eq:primal-sdp}
\minimize_{z,u,t}~ &  \frac{1}{2}\vnorm{y - z}_2^2 + \frac{\tau}{2}(t + u_1)\\
\operatorname{subject\ to\ } & \begin{pmatrix}
\operatorname{Toep}(u) & z\\
z^* & t
\end{pmatrix} \succeq 0.
\end{align}
where $\operatorname{Toep}(u)$ denotes a Hermitian Toeplitz matrix with $u$ as
its first row and $u_1$ is the first component of $u$. Similarly, the dual
semi-infinite program \eqref{atnorm-dual} is equivalent to the dual
semidefinite program of~\eqref{eq:primal-sdp}.
\end{remark}

\section{What is the best rate we can expect?}\label{sec:minimax}
Using results about minimax achievable rates for linear models~
\cite{cd_minimax,rw_minimax}, we can deduce that the convergence rate stated in
\eqref{fast-rate} is near optimal. Define the set of $k$ well separated
frequencies as
\[
\mathcal{S}_k = \left\{(f_1, \dots, f_k) \in \mathbb{T}^k ~\middle|~  d(f_p, 
f_q) \geq 4/n, p \neq q \right\}
\]
The expected minimax denoising error $M_k$ for a line spectral signal with
frequencies from $\mathcal{S}_k$ is defined as the lowest expected denoising
error rate for any estimate $\hat{x}(y)$ for the worst case signal $x^\star$
with support $T(x^\star) \in \mathcal{S}_k$. Note that we can lower bound $M_k$
by restricting the set of candidate frequencies to smaller set. To that end,
suppose we restrict the signal $x^\star$ to have frequencies only drawn from an
equispaced grid on the torus $T_n := \{ 4 j/n \}_{j=1}^{n/4}$. Note that any
set of $k$ frequencies from $T_n$ are pairwise separated by at least $4/n$. If
we denote by $F_n$ a $n \times (n/4)$ partial DFT matrix with (unnormalized)
columns corresponding to frequencies from $T_n$, we can write $x^\star = F_n
c^\star$ for some $c^\star$ with $\vnorm{c^\star}_0 = k$. Thus,
\begin{align*}
M_k &:= \inf_{\hat{x}}
 \sup_{
	T(x^\star) \in \mathcal{S}_k}
\frac{1}{n} \mathbb{E} \vnorm{\hat{x} - x^\star}_2^2
	\\
&\geq \inf_{\hat{x}} 
 \sup_{
	\vnorm{c^\star}_0 \leq k
	} \frac{1}{n} \mathbb{E} \vnorm{\hat{x} - F_n c^\star}_2^2\\
&\geq \inf_{\hat{c}}
 \sup_{\vnorm{c^\star}_0 \leq k} \frac{1}{n} \mathbb{E} \vnorm{F_n(\hat{c} - 
 c^\star)}_2^2\\
&\geq  \frac{n}{4} \left\{ \inf_{\hat{c}}
 \sup_{\vnorm{c^\star}_0 \leq k}\frac{4}{n} \mathbb{E} \vnorm{\hat{c} - 
 c^\star}_2^2\right\}\,.
\end{align*}
Here, the first inequality is the restriction of $T(x^\star)$. The second
inequality follows because we project out all components of $\hat{x}$ that do
not lie in the span of $F_n$. Such projections can only reduce the Euclidean
norm. The third inequality uses the fact that the minimum singular value of
$F_n$ is $n$ since $F_n^*F_n = n I_{{n}/{4}}$. Now we may directly apply the
lower bound for estimation error for linear models derived by Cand\'es and
Davenport. Namely, Theorem 1 of~\cite{cd_minimax} states that
\begin{align*}
\inf_{\hat{c}}
 \sup_{\vnorm{c^\star}_0 \leq k} \frac{4}{n} \mathbb{E} \vnorm{\hat{c} - 
 c^\star}_2^2&\geq {C} \sigma^2 \frac{k 
 \log\left(\frac{n}{4k}\right)}{\vnorm{F_n}_\mathrm{F}^2}\,.
 \end{align*} With the preceding analysis and the fact that 
 $\vnorm{F_n}_{\mathrm{F}}^2 = n^2/4$, we can thus deduce the following theorem:
 \begin{theorem}
\label{minimax}
Let $x^\star$ be a line spectral signal as described by \eqref{signal} with the
support $T(x^\star) = \{f_1, \dots, f_k\} \in \mathcal{S}_k$ and $y = x^\star +
w$, where $w \in \C^n$ is circularly symmetric Gaussian noise with variance
$\sigma^2 I_n$. Let $\hat{x}$ be any estimate of $x^\star$ using $y$. Then,
\[
M_k = \inf_{\hat{x}}
 \sup_{
	T(x^\star) \in \mathcal{S}_k}
\frac{1}{n} \mathbb{E} \vnorm{\hat{x} - x^\star}_2^2
\geq C\sigma^2 \frac{k \log\left(\frac{n}{4k}\right)}{n}
\]
for some constant $C$ that is independent of $k$, $n$, and $\sigma$.
\end{theorem}

This theorem and Theorem~\ref{main} certify that AST is nearly minimax optimal
for spectral estimation of well separated frequencies.

\section{Proofs of Main Theorems}
\label{sec:proofs}

In this section, there are many numerical constants. Unless otherwise
specified, $C$ will denote a numerical constant whose value may change from
equation to equation. Specific constants will be highlighted by accents or
subscripts.

We describe the preliminaries and notations, and restate some recent results we
used before sketching the proof of Theorems \ref{main} and \ref{support}.

\subsection{Preliminaries}
The sample $x^\star_j$ may be regarded as the $j$th trigonometric moment of 
the discrete measure $\mu$ given by \eqref{mu}:
\begin{eqnarray*}
  x_j^\star & = & \int_0^1 e^{i 2 \pi j f} \mu ( d f)
\end{eqnarray*}
for $-m \leq j \leq m$.
Thus, the problem of extracting the frequencies and amplitudes from noisy
observations may be regarded as the inverse problem of estimating a measure
from noisy trigonometric moments.

We can write the vector $x^\star$ of observations $[x_{-m}^\star, \ldots, 
x_m^\star]^T$ in terms of an \emph{atomic decomposition}
\[
x^\star = \sum_{l=1}^k c_l a(f_l)
\]
or equivalently in terms of a corresponding \emph{representing measure} $\mu$ 
given by \eqref{mu} satisfying
\[
x^\star = \int_0^1 a(f) \mu(df)
\]
There is a one-one correspondence between atomic decompositions and
representing measures. Note that there are infinite atomic decompositions of
$x^\star$ and also infinite corresponding representing measures. However, since
every collection of $n$ atoms is linearly independent, $\A$ forms a full spark
frame~\cite{spark} and therefore the problem of finding the sparsest
decomposition of $x^\star$ is well-posed if there is a decomposition which is
at least $n/2$ sparse.

The atomic norm of a vector $z$ defined in \eqref{def-atnorm} is the minimum
total variation norm~\cite{cs_otg,tvnorm} $\vnorm{\mu}_{\mathrm{TV}}$ of all
representing measures $\mu$ of $z$. So, minimizing the total variation norm is
the same as finding a decomposition that achieves the atomic norm.

\subsection{Dual Certificate and Exact Recovery}

Atomic norm minimization attempts to recover the sparsest decomposition by
finding a decomposition that achieves the atomic norm, i.e., find ${c_l,f_l}$
such that $x^\star = \sum_l c_l a(f_l)$ and $ \vnorm{x^\star}_\A = \sum_l |c_l|
$ or equivalently, finding a representing measure $\mu$ of the form \eqref{mu}
that minimizes the total variation norm $ \vnorm{\mu}_{\mathrm{TV}}$. The
authors of~\cite{cg_exact12} showed that when $n > 256$, the decomposition that
achieves the atomic norm is the sparsest decomposition by explicitly
constructing a dual certificate~\cite{dualcert} of optimality, whenever the
composing frequencies $f_1, \ldots, f_k$ satisfy a minimum separation
condition~\eqref{min-sep}. In the rest of the paper, we always make the
technical assumption that $n > 256$.

\begin{definition}[Dual Certificate]
\label{dual-cert}
A vector $q \in \C^n$ is called a dual certificate for $x^\star$ if for the
corresponding trigonometric polynomial $Q(f) := \langle q, a(f) \rangle$, we
have $$Q(f_l) = \operatorname{sign}(c_l), l = 1, \ldots, k$$ and $$|Q(f)| < 1$$
whenever $f\not\in \{ f_1, \ldots, f_k\}$.
\end{definition}
The authors of ~\cite{cg_exact12} not only explicitly constructed such a
certificate characterized by the dual polynomial $Q$, but also showed that
their construction satisfies some stability conditions, which is crucial for
showing that denoising using the atomic norm provides stable recovery in the
presence of noise.

\begin{theorem}[Dual Polynomial Stability, Lemma 2.4 and 2.5 in
\cite{cg_noisy}] \label{dual-stab} For any $f_1, \ldots, f_k$ satisfying the
separation condition \eqref{min-sep} and any sign vector $v \in \C^k$ with
$|v_j|=1$, there exists a trigonometric polynomial $Q = \left<q, a(f)\right>$
for some $q \in \C^n$ with the following properties:
\begin{enumerate}
\item For each $j = 1, \ldots, k$, $Q$ interpolates the sign vector $v$ so that 
$Q(f_j) = v_j$
\item In each neighborhood $N_j$ corresponding to $f_j$ defined by
$N_j = \left\{ f : d(f, f_j) < {0.16}/{n} \right\}$, 
the polynomial $Q(f)$ behaves like a quadratic and there exist constants $C_a, 
C_a'$ so that
\begin{align}
\label{q1}|Q(f)| & \leq 1 - \frac{C_a}{2} n^2 (f-f_j)^2\\
\label{q2}|Q(f) - v_j| & \leq \frac{C_a'}{2} n^2 (f - f_j)^2
\end{align}
\item When $f \in F = [0,1] \backslash \cup_{j=1}^k{N_j}$, there is a numerical 
constant $C_b>0$ such that
\[
|Q(f)| \leq 1 - C_b
\]
\end{enumerate}
\end{theorem}

We use results in~\cite{cg_noisy} and~\cite{btr12} (reproduced in Appendix D
for convenience) and borrow several ideas from the proofs in~\cite{cg_noisy},
with nontrivial modifications to establish the error rate of atomic norm
regularization.

\subsection{Proof of Theorem~\ref{main}}
Let $\hat{\mu}$ be the representing measure for the solution $\hat{x}$  of 
\eqref{atnorm-denoise} with minimum total variation norm, that is,
\[
\hat{x} = \int_0^1 a(f) \hat{\mu}(df)
\]
and $\vnorm{\hat{x}}_\A = \vnorm{\hat{\mu}}_{\mathrm{TV}}$. Denote the error
vector by $e = x^\star - \hat{x}$. Then, the difference measure $\nu = \mu -
\hat{\mu}$ is a representing measure for $e$. We first express the denoising
error $\vnorm{e}_2^2$ as the integral of the error function $E(f) = \langle e,
a(f) \rangle,$ against the difference measure $\nu$:
\begin{align*}
\vnorm{e}_2^2 &= \langle e, e \rangle\\
& = \left\langle e, \int_0^1 a(f) \nu(df) \right\rangle\\
& =  \int_0^1  \left\langle e,a(f) \right\rangle \nu(df)\\
& = \int_0^1 E(f) \nu(df).
\end{align*}

Using a Taylor series approximation in each of the near regions $N_j$, we first
show that the denoising error (or in general any integral of a trigonometric
polynomial against the difference measure) can be controlled in terms of an
integral in the far region $F$ and the zeroth, first, and second moments of the
difference measure in the near regions. The precise result is presented in the
following lemma, whose proof is given in Appendix
\ref{apx:pf:taylor}.
\begin{lemma}
\label{part1}
Define
\begin{align*} 
I_0^j &:= \left| \int_{N_j} \nu(df) \right|\\
I_1^j &:= n \left| \int_{N_j} (f-f_j) \nu(df) \right|\\
I_2^j &:= \frac{n^2}{2} \int_{N_j} (f-f_j)^2 |\nu|(df)\\
I_l &:= \sum_{j=1}^k I_l^j,~~\mbox{for}~l = 0, 1, 2\,.
\end{align*}
Then for any $m$th order trigonometric polynomial $X$, we have
\[
\int_0^1{ X(f) \nu(df)}
\leq \vnorm{X(f)}_\infty \left(\int_F{|\nu|(df)} + I_0 + I_1 + I_2\right)
\]
\end{lemma}

Applying Lemma \ref{part1} to the error function, we get
\begin{equation}
\label{ebd}
\vnorm{e}_2^2 \leq \vnorm{E(f)}_\infty 
\left( \int_F{|\nu|(df)} + I_0 + I_1 + I_2\right)
\end{equation}
As a consequence of our choice of $\tau$ in \eqref{tau}, we can show that
$\vnorm{E(f)}_\infty \leq (1+2\eta^{-1})\tau$ with high probability. In fact,
we have
\begin{align*}
\vnorm{E(f)}_\infty &= \sup_{f \in [0,1]}\left|\langle e, a(f) \rangle\right|\\
&= \sup_{f \in [0,1]} \left| \langle x^\star - \hat{x}, a(f) \rangle\right|\\
&\leq \sup_{f \in [0,1]} \left| \langle w, a(f) \rangle \right| +  \sup_{f \in [0,1]} \left| \langle y - \hat{x}, a(f) \rangle \right|\\
&\leq \sup_{f \in [0,1]} \left| \langle w, a(f) \rangle \right| +  \tau\\
\label{errbd} \numberthis &\leq (1 +2\eta^{-1})\tau \leq 3 \tau, \text{with 
high probability.}
\end{align*}
The second inequality follows from the optimality conditions for 
\eqref{atnorm-denoise}. It is shown in Appendix C of ~\cite{btr12} that the 
penultimate inequality holds with high probability.

Therefore, to complete the proof, it suffices to show that the other terms on
the right hand side of \eqref{ebd} are $O(\frac{k\tau}{n})$. While there is no
exact frequency recovery in the presence of noise, we can hope to get the
frequencies approximately right. Hence, we expect that the integral in the far
region can be well controlled and the local integrals of the difference measure
in the near regions are also small due to cancellations. Next, we utilize the
properties of the dual polynomial in Theorems~\ref{dual-stab} and another
polynomial given in Theorem~\ref{dual-lin} in Appendix \ref{apx:collection} to
show that the zeroth and first moments of $\nu$ may be controlled in terms of
the other two quantities in \eqref{ebd} to upper bound the error rate. The
following lemma is similar to Lemmas 2.2 and 2.3 in~\cite{cg_noisy}, but we
have made several modifications to adapt it to our signal and noise model. For
completeness, we provide the proof in Appendix \ref{apx:pf:I0I1}.

\begin{lemma}
\label{part2}
There exists numeric constants $C_0$ and $C_1$ such that
\begin{align*}
I_0 &\leq C_0 \left(\frac{k \tau}{n} + I_2 + \int_F{|\nu|(df)}\right) \\
I_1 &\leq C_1 \left(\frac{k \tau}{n} + I_2 + \int_F{|\nu|(df)}\right).
\end{align*}
\end{lemma}

All that remains to complete the proof is an upper bound on $I_2$ and $\int_F{|\nu|(df)}$.  The key idea in establishing such a bound is deriving upper and lower bounds on the
difference $\| P_{T^c} ( \nu) \|_{{\mathrm{TV}}} - \| P_T ( \nu) \|_{{\mathrm{TV}}}$
between the total variation norms of $\nu$ on and off the support. The upper bound can be derived using optimality conditions. We lower bound $\| P_{T^c} ( \nu)
\|_{{\mathrm{TV}}} - \| P_{T} ( \nu) \|_{{\mathrm{TV}}}$ using the fact that a constructed dual
certificate $Q$ has unit magnitude for every element in the support
$T$ of $P_T ( \nu)$ whence we have $\| P_T ( \nu) \|_{{\mathrm{TV}}} = \int_{\mathbb{T}}
Q ( f) \nu ( d f)$. A critical element in deriving both the lower and upper bounds is that the dual polynomial $Q$ has quadratic drop in each near regions $N_j$ and is bounded away from one in the far region $F$. Finally, by combing these bounds and carefully controlling the regularization parameter, we get the desired result summarized in the following lemma. The details of the proof are fairly technical and we leave them to Appendix \ref{apx:pf:I2far}.

\begin{lemma}
Let $\tau = \eta\sigma \sqrt{n\log(n)}$. If $\eta>1$ is large enough, then there exists a numerical constant $C$ such that, with high probability
\label{part3}
\[
\int_F{|\nu|(df)} + I_2 \leq \frac{C k \tau}{n}.
\]
\end{lemma}

Putting together Lemmas \ref{part1}, \ref{part2} and \ref{part3}, we finally prove our main theorem:
\begin{align*}
\frac{1}{n}\vnorm{e}_2^2 
&\leq \frac{\vnorm{E(f)}_\infty}{n} \left(\int_F{|\nu|(df)} + I_0 + I_1 + I_2\right)\\
&\leq \frac{\vnorm{E(f)}_\infty}{n} \left(\frac{C_1 k \tau}{n} + C_2 \int_F{|\nu|(df)} + C_3 I_2\right)\\
&\leq  \frac{\vnorm{E(f)}_\infty}{n} \frac{C k \tau}{n} \\
& \leq \frac{C k \tau^2}{n^2}\\
&= O\left(\sigma^2\frac{k \log(n)}{n}\right).
\end{align*}
The first three inequalities come from successive applications of Lemmas 1, 2 and 3 respectively. The fourth inequality follows from \eqref{errbd} and the fifth by our choice of $\tau$ according to Eq. \eqref{tau}. This completes the proof of Theorem \ref{main}.

\subsection{Proof of Theorem~\ref{support}}
\label{sec:support}
The first two statements in Theorem \ref{support} are direct consequences of Lemma~\ref{part3}. For (iii.), we follow~\cite{granda2} and  use the dual polynomial $Q_j^{\star} ( f) = \langle q_j^{\star}, a ( f)\rangle$ constructed in Lemma 2.2 of~\cite{granda2} which satisfies
\begin{eqnarray*}
  Q_j^{\star} ( f_j) & = & 1\\
  | 1 - Q_j^{\star} ( f) | & \leq & n^2 C_1' ( f - f_j)^2, f \in N_j\\
  | Q_j^{\star} ( f) | & \leq & n^2 C_1' ( f - f_{j'})^2, f \in N_{j'}, j' \neq
  j\\
  | Q_j^{\star} ( f) | & \leq & C_2', f \in F.
\end{eqnarray*}
We note that $c_j - \sum_{\hat{f}_l \in N_j} \hat{c}_l = \int_{N_j} \nu(df)$. Then, by applying triangle inequality several times,
\begin{align*}
\left| \int_{N_j}  \nu(df)\right|
& \leq \left| \int_{N_j}  Q_j^\star (f) \nu(df)\right| + \left| \int_{N_j}  (1-Q_j^\star (f)) \nu(df)\right|\\
& \leq \left| \int_0^1  Q_j^\star (f) \nu(df)\right| + \left| \int_{N_j^c}  Q_j^\star (f) \nu(df)\right| + \left| \int_{N_j}  (1-Q_j^\star (f)) \nu(df)\right|\\
& \leq \left|\int_0^1  Q_j^\star (f) \nu(df)\right| + \left| \int_{F}  Q_j^\star (f) \nu(df)\right| \\
&\qquad\qquad\qquad + \sum_{\substack{j' \neq j\\j'=1}}^k \int_{N_{j'}} \left| Q_j^\star (f)\right| |\nu|(df) +  \int_{N_j}  \left|1-Q_j^\star (f)\right| |\nu(df)|\,.
\end{align*}

We upper bound the first term using Lemma~\ref{l4} in Appendix \ref{apx:collection} which yields
\[
\left| \int^0_{1}  Q_j^\star (f) \nu(df)\right| \leq \frac{Ck \tau}{n}
\]
The other terms can be controlled using the properties of $Q_j^\star$:
\begin{align*}
\left| \int_{F}  Q_j^\star (f) \nu(df)\right| & \leq C_2' \int_{F} |\nu| (df)\\
\sum_{\substack{j' \neq j\\j'=1}}^k \int_{N_{j'}} \left| Q_j^\star (f)\right| |\nu|(df) +  \int_{N_j}  \left|1-Q_j^\star (f)\right| |\nu|(df)
& \leq
 C_1'\sum_{j'=1}^k \int_{N_{j'}} n^2 (f-f_{j'})^2 |\nu|(df) = C_1 I_2
\end{align*}
Using Lemma~\ref{part3}, both of the above are upper bounded by $\frac{C k \tau}{n}$. Now, by combining these upper bounds, we finally have
\[
\left| c_j - \sum_{l : \hat{f}_l \in N_j} \hat{c}_l \right| \leq \frac{C_3 k \tau}{n}
\]
This shows part (iii) of the theorem. Part (iv) can be obtained by combining parts (ii) and (iii).

\section{Experiments}
\label{sec:experiments}

In~\cite{btr12}, we demonstrated with extensive experiments that AST
outperforms classical subspace algorithms in terms of mean squared estimation error.  In the experiments here, we focus on frequency localization and 
compare the  performance of AST, MUSIC~\cite{music} and
Cadzow's method~\cite{cadzow05} under various choices of number of frequencies,
number of samples and signal to noise ratios (SNRs).

We adopt the same experimental setup as in~\cite{btr12} and reproduce the
description of experiments here for convenience. We generated $k$ normalized
frequencies $f_1, \ldots, f_k$ uniformly randomly chosen from $\left[0,1\right]$
such that every pair of frequencies are separated by at least $1/2n$. The signal
$x^\star \in \C^n$ is generated according to \eqref{signal} with $k$ random
amplitudes independently chosen from $\chi^2(1)$ distribution (squared
Gaussian). All of our sinusoids were then assigned a random phase (equivalent to
multiplying $c_l$ by a random unit norm complex number). The observation
$y$ is produced by adding complex white gaussian noise $w$ such that the input
signal to noise ratio (SNR) is $-10,-5,0,5,10,15$ or $20$ dB. We compared the
average value of the following metrics of the various algorithms in 20 random
trials for various values of number of observations $(n = 64,128,256)$, and
number of frequencies $ (k = n/4,n/8,n/16)$.

AST needs an estimate of the noise variance $\sigma^2$ to pick the
regularization parameter according to \eqref{tau}. In our experiments, we do not provide our algorithm with the true noise variance.  Instead, we can construct an estimate for $\sigma$ with 
the following heuristic.   We formed the empirical autocorrelation
matrix using the MATLAB routine {\tt corrmtx} using a prediction order $n/3$
and averaging the lower $25\%$ of the eigenvalues. We then use this estimate in
equation~\eqref{tau} to determine the regularization parameter.   See~\cite{btr12} for more details.

We implemented AST using the Alternating Direction Method of Multipliers
(ADMM, see for example,~\cite{admm2011}, or~\cite{btr12} for the specific details).
We used the stopping criteria described in~\cite{admm2011} and set $\rho=2$ for
all experiments. We use the dual solution $\hat{z}$ to determine the support of
the optimal solution $\hat{x}$. Once the frequencies $\hat{f}_l$ are extracted,
we ran the least squares problem $\mbox{minimize}_\alpha \|U \alpha - y\|^2$
where $U_{jl} = \exp(i 2\pi j \hat{f}_l)$ to obtain \emph{debiased} estimates of the
amplitudes.

We implemented Cadzow's method as described by the pseudocode in~\cite{blu08},
and MUSIC~\cite{music} using the MATLAB routine {\tt rootmusic}. These
algorithms need an estimate of the number of sinusoids. Rather than implementing
a heuristic to estimate $k$, \emph{we fed the true $k$ to our solvers}. This
provides a significant advantage to these algorithms. On the contrary, AST is
not provided the true value of $k$, and the noise variance $\sigma^2$ required
in the regularization parameter is estimated from $y$.

Let $\{\hat{c}_l\}$ and $\{\hat{f}_l\}$ denote the amplitudes and frequencies
estimated by any of the algorithms - AST, MUSIC or Cadzow.  We use the
following error metrics to characterize the frequency localization of various
algorithms:
\begin{enumerate}
\item[(i)] Sum of the absolute value of amplitudes in the far region $F$, $m_1 = \sum_{l : \hat{f}_l \in F} |\hat{c}_l|$
\item[(ii)] The weighted frequency localization error, $m_2 = \sum_{l : \hat{f}_l \in N_j} |\hat{c}_l| \{ \min_{f_j \in T} d(f_j,\hat{f}_l) \}^2$
\item[(iii)] Error in approximation of amplitudes in the near region, $m_3 = \left| c_j - \sum_{l : \hat{f_l} \in N_j} \hat{c}_l \right|$
\end{enumerate}
These are precisely the quantities that we prove tend to zero in Theorem~\ref{support}.

To summarize the results, we first provide \emph{performance profiles} to summarize the behavior of the various
algorithms across all of the parameter settings. Performance profiles provide a
good visual indicator of the relative performance of many algorithms under a
variety of experimental conditions~\cite{dolanmore02}. Let $\mathcal{P}$ be the
set of experiments and let $e_s(p)$ be the value of the error measure $e$ of
experiment $p \in \mathcal{P}$ using the algorithm $s$. Then the ordinate
$P_s(\beta)$ of the graph at $\beta$ specifies the fraction of experiments
where the ratio of the performance of the algorithm $s$ to the minimum error
$e$ across all algorithms for the given experiment is less than $\beta$, i.e.,
\begin{equation*}
P_s(\beta) = \frac{\mathop{\#}\left\{p \in \mathcal{P} ~:~ e_s(p) \leq \beta \min_s e_s(p)\right\}}{\mathop{\#}(\mathcal{P})}
\end{equation*}

\begin{figure}[htp]
\begin{tabular}{ccc}
\includegraphics[height=42mm]{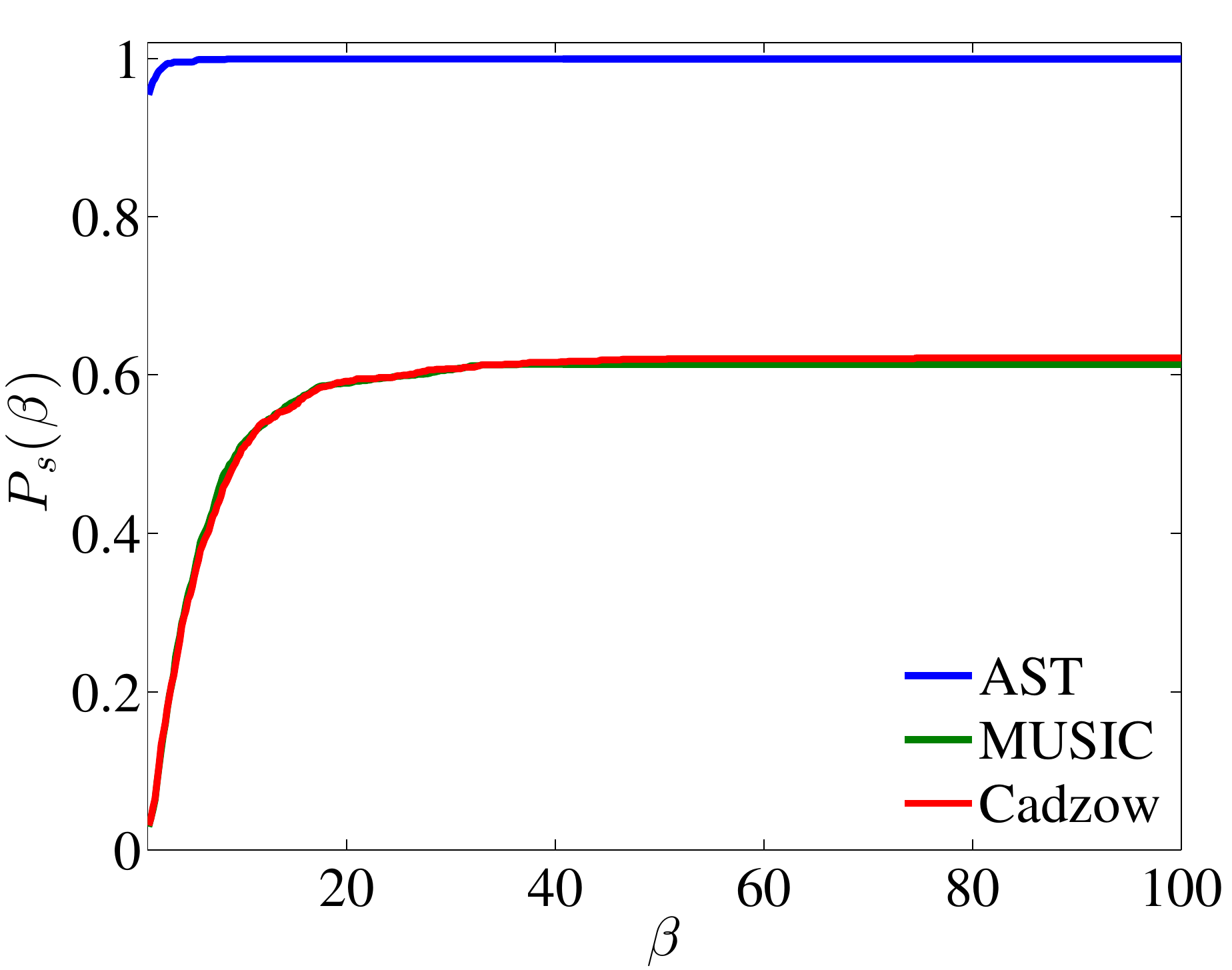} &
\includegraphics[height=42mm]{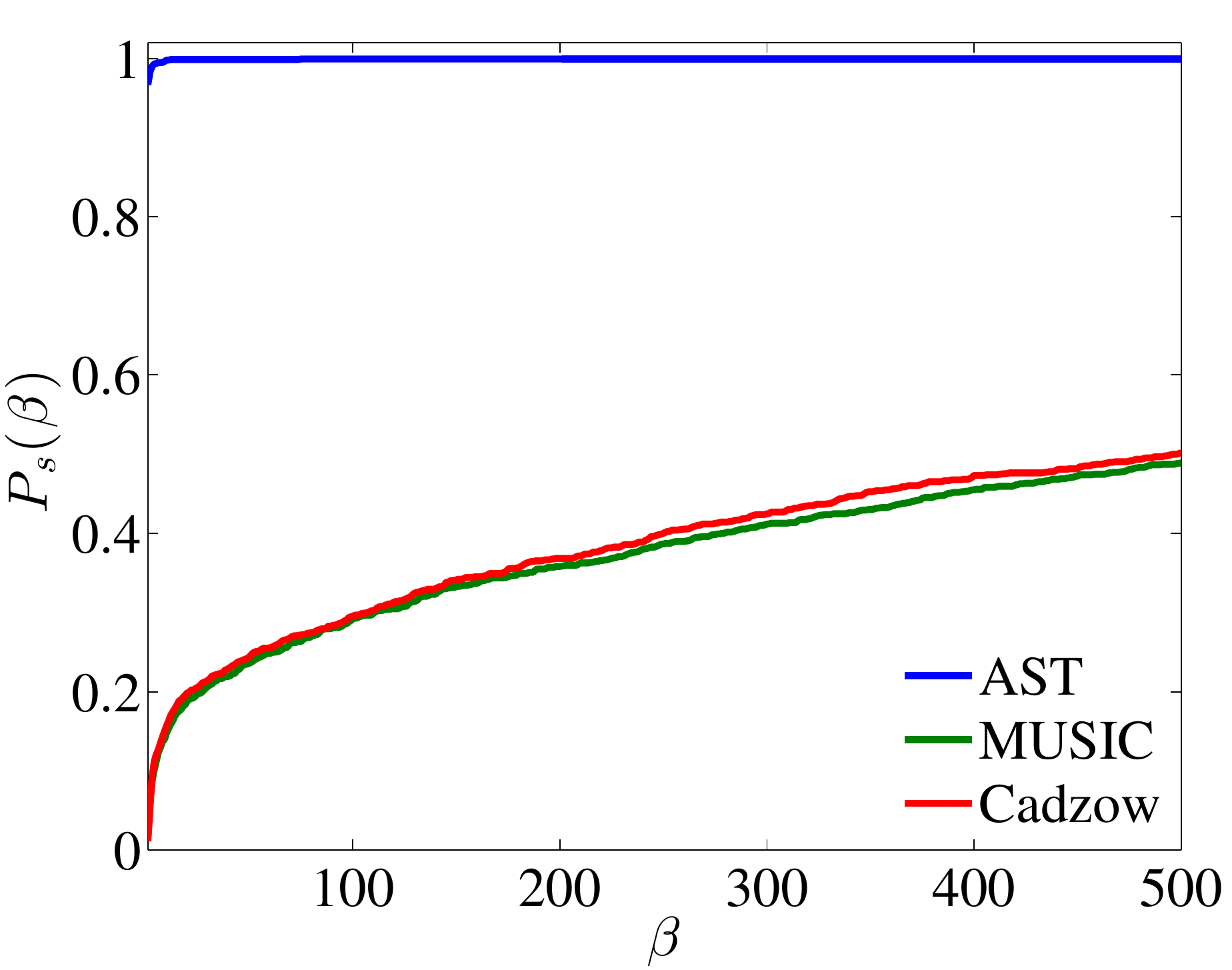} &
\includegraphics[height=42mm]{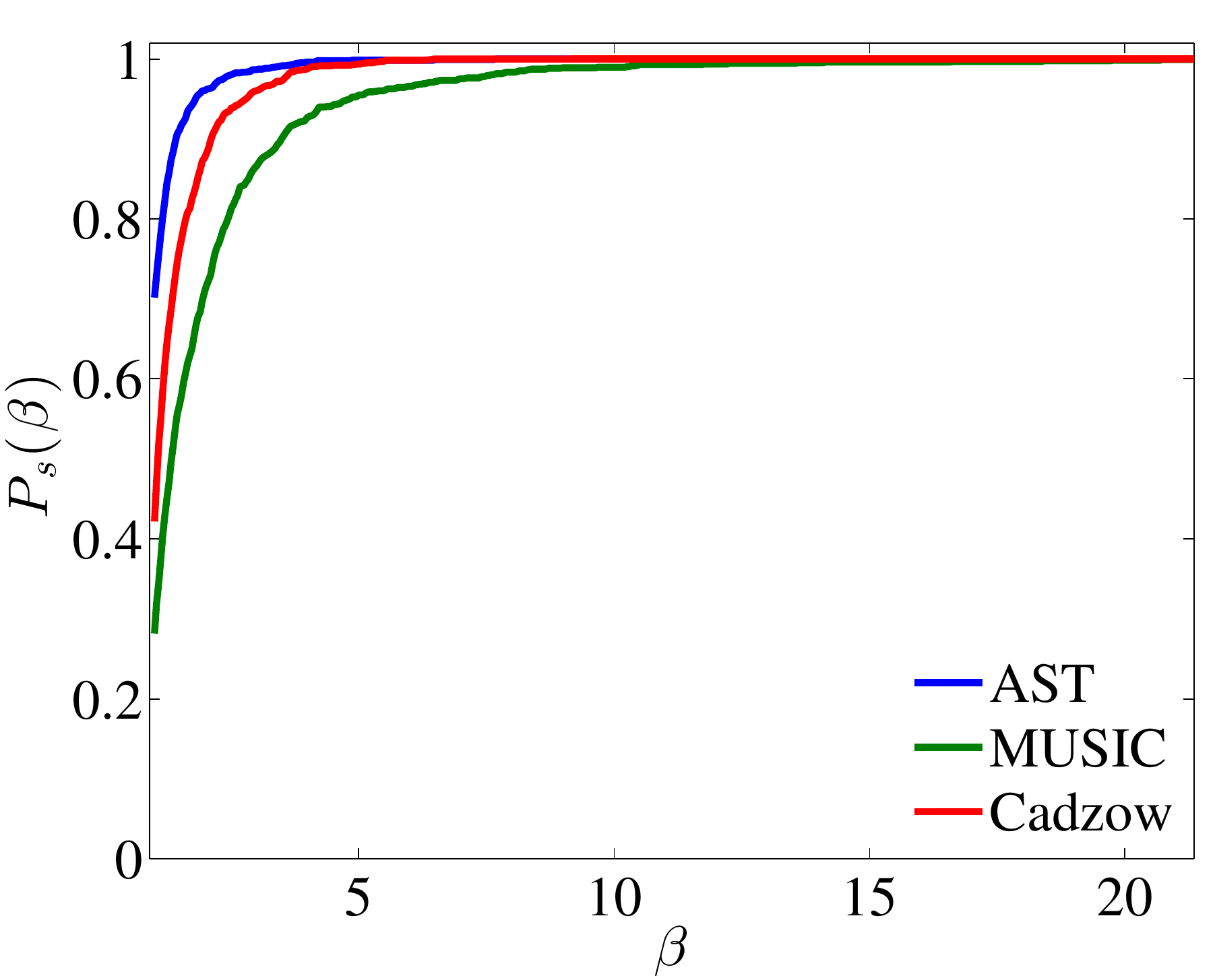}\\
(a) $m_1$ & (b) $m_2$ & (c) $m_3$
\end{tabular}
\caption{ Performance Profiles for AST, MUSIC and Cadzow.
(a) Sum of the absolute value of amplitudes in the far region ($m_1$)
(b) The weighted frequency localization error, $m_2$
(c) Error in approximation of amplitudes in the near region, $m_3$ }
\label{fig:pp}
\end{figure}

The performance profiles in Figure~\ref{fig:pp} show that AST is the best
performing algorithm for all the three metrics.  AST in fact outperforms
MUSIC and Cadzow by a substantial margin for metrics $m_1$ and $m_2$. 

In Figure~\ref{fig:msnr}, we display how the error metrics vary with
increasing SNR for AST, MUSIC and Cadzow.  We restrict these plots to the experiments with $n =
256$ samples. These plots demonstrate that AST localizes frequencies
substantially better than MUSIC and Cadzow even for low signal to noise ratios
as there is very little energy in the far region of the frequencies ($m_1$) and
has the smallest weighted mean square frequency deviation ($m_2$). Although we
have plotted the average value in these plots, we observed spikes in the plots
for Cadzow's algorithm as the average is dominated by the worst performing
instances.  These large errors are due to the numerical instability of polynomial root finding. 

\begin{figure}[htbp]
\begin{tabular}{ccc}
\includegraphics[height=43mm]{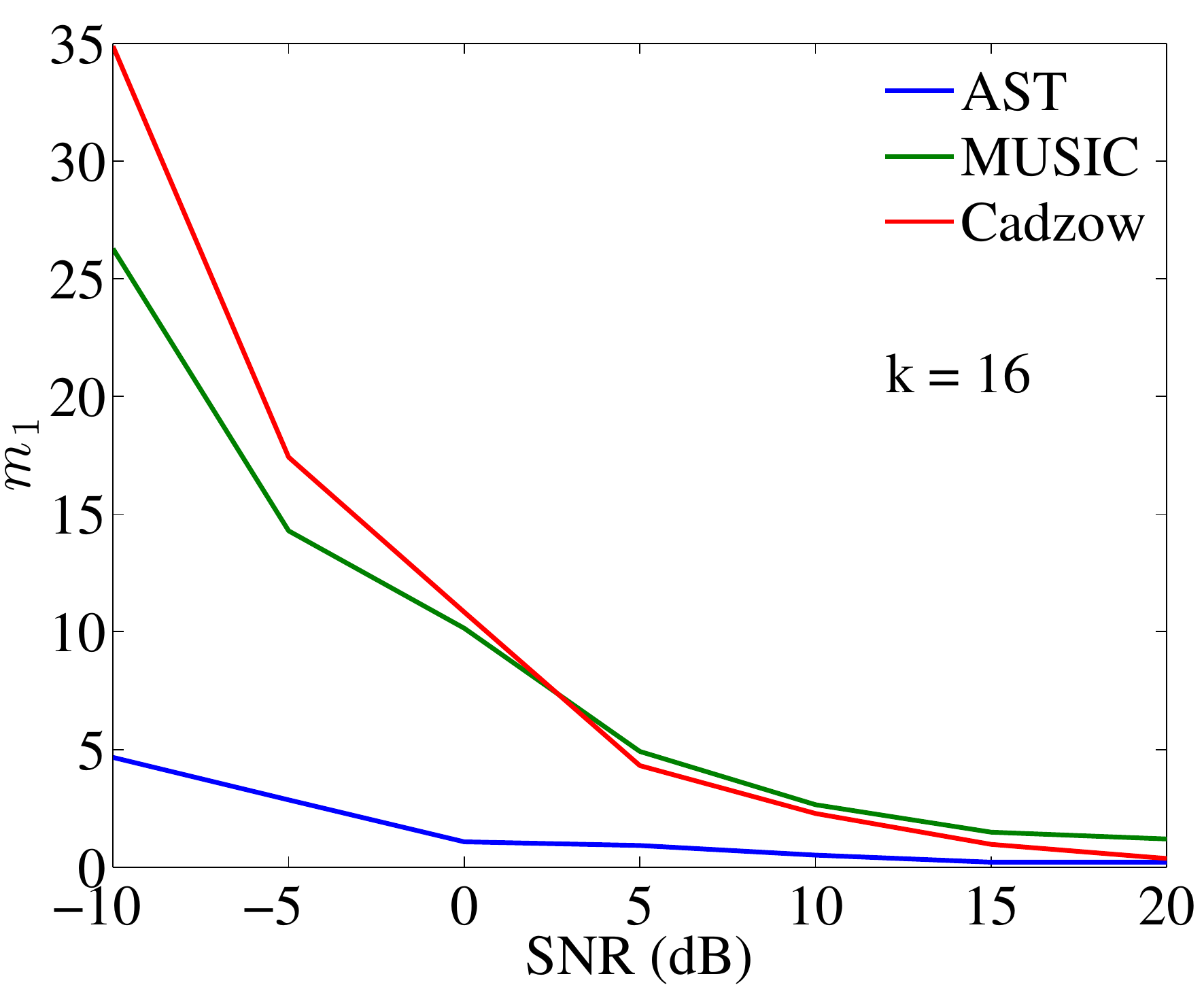} &
\includegraphics[height=43mm]{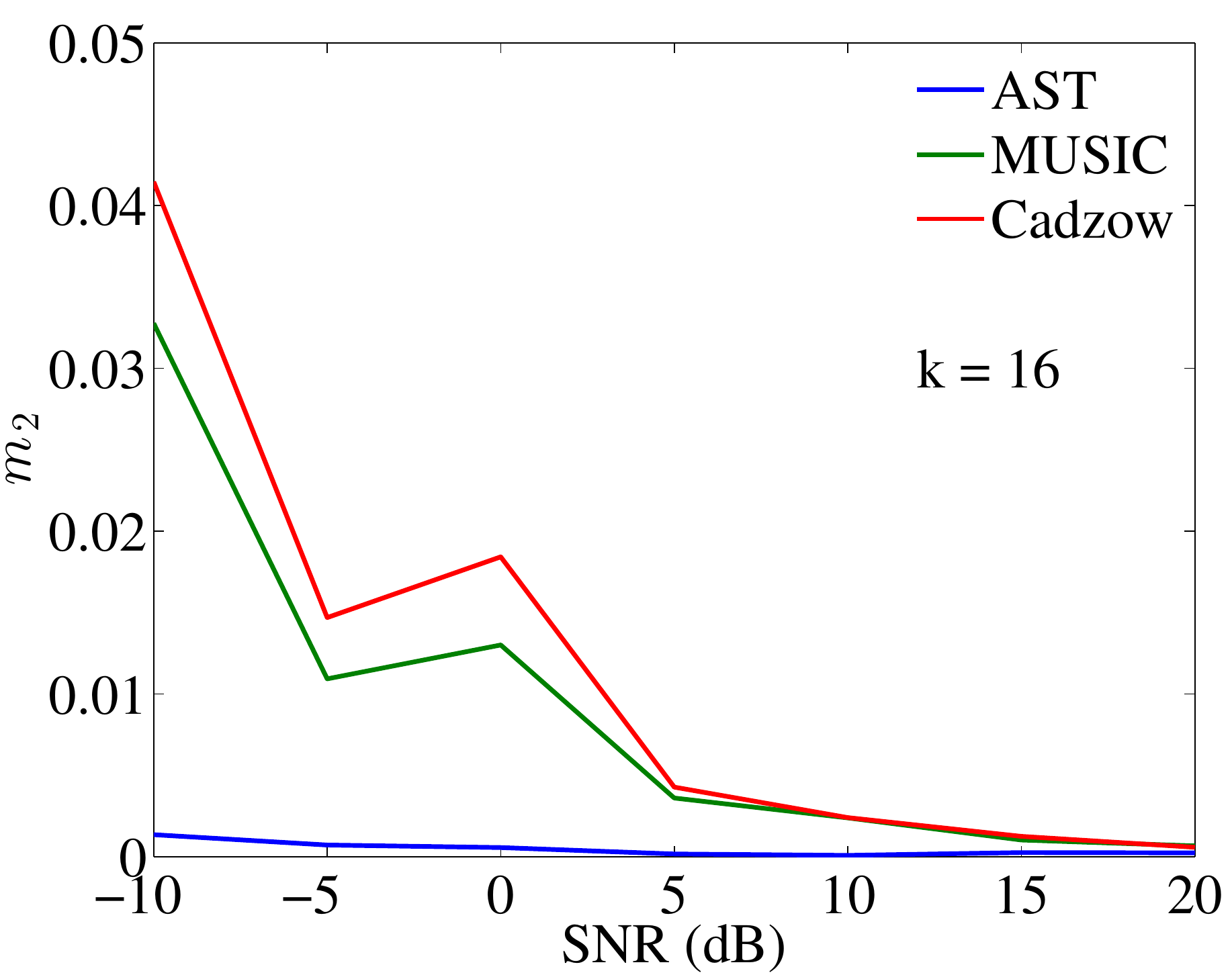} &
\includegraphics[height=43mm]{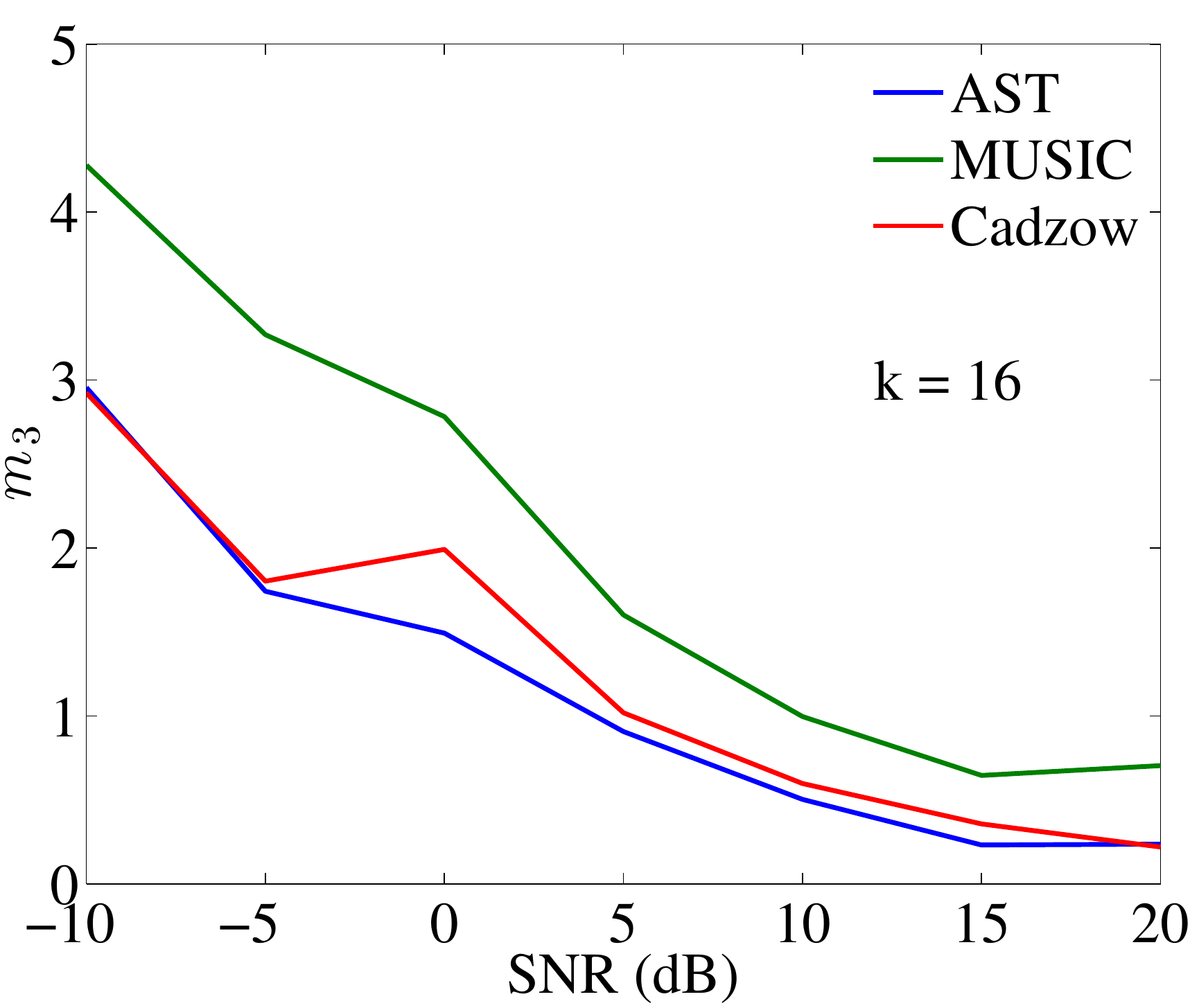} \\
\includegraphics[height=43mm]{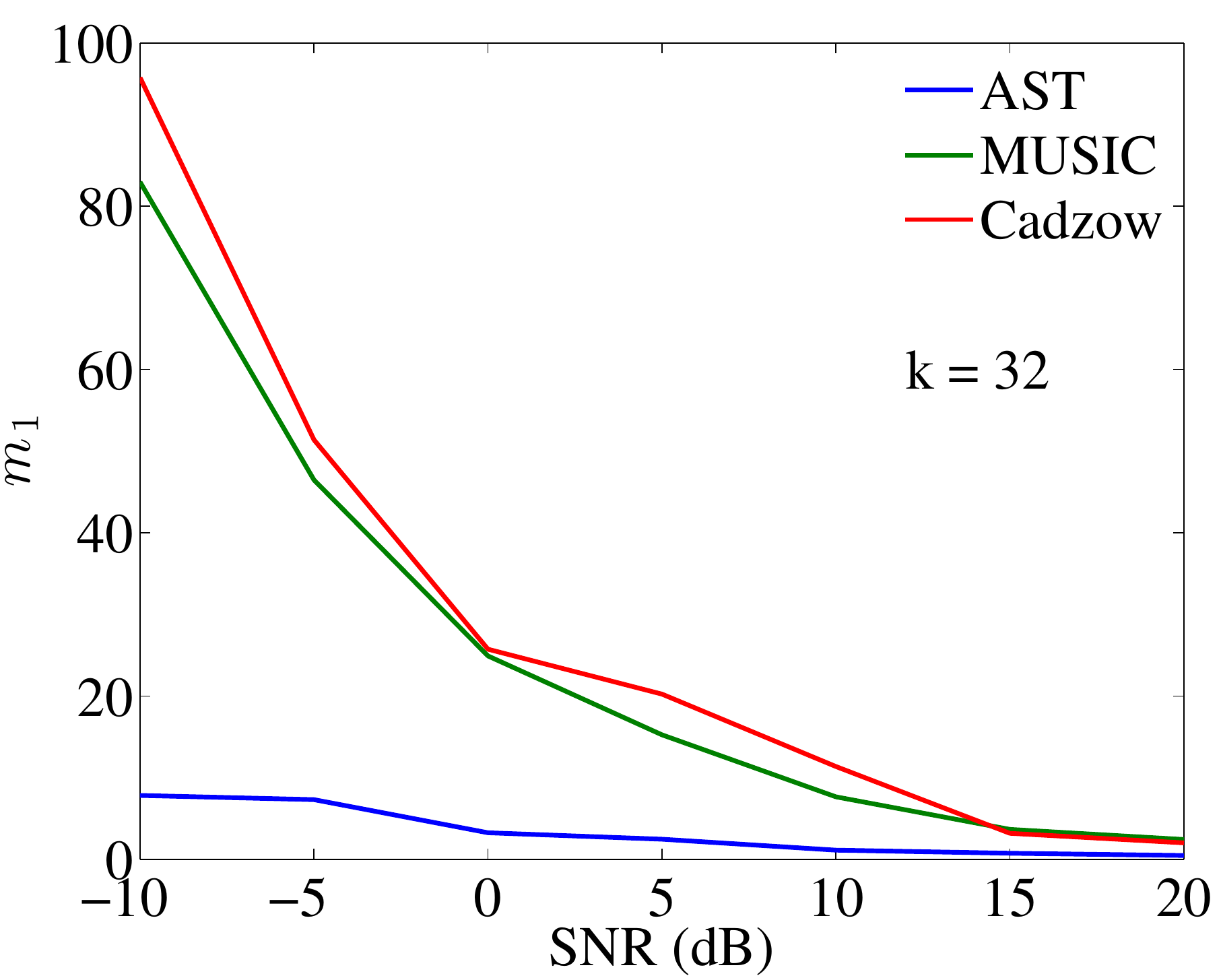} &
\includegraphics[height=43mm]{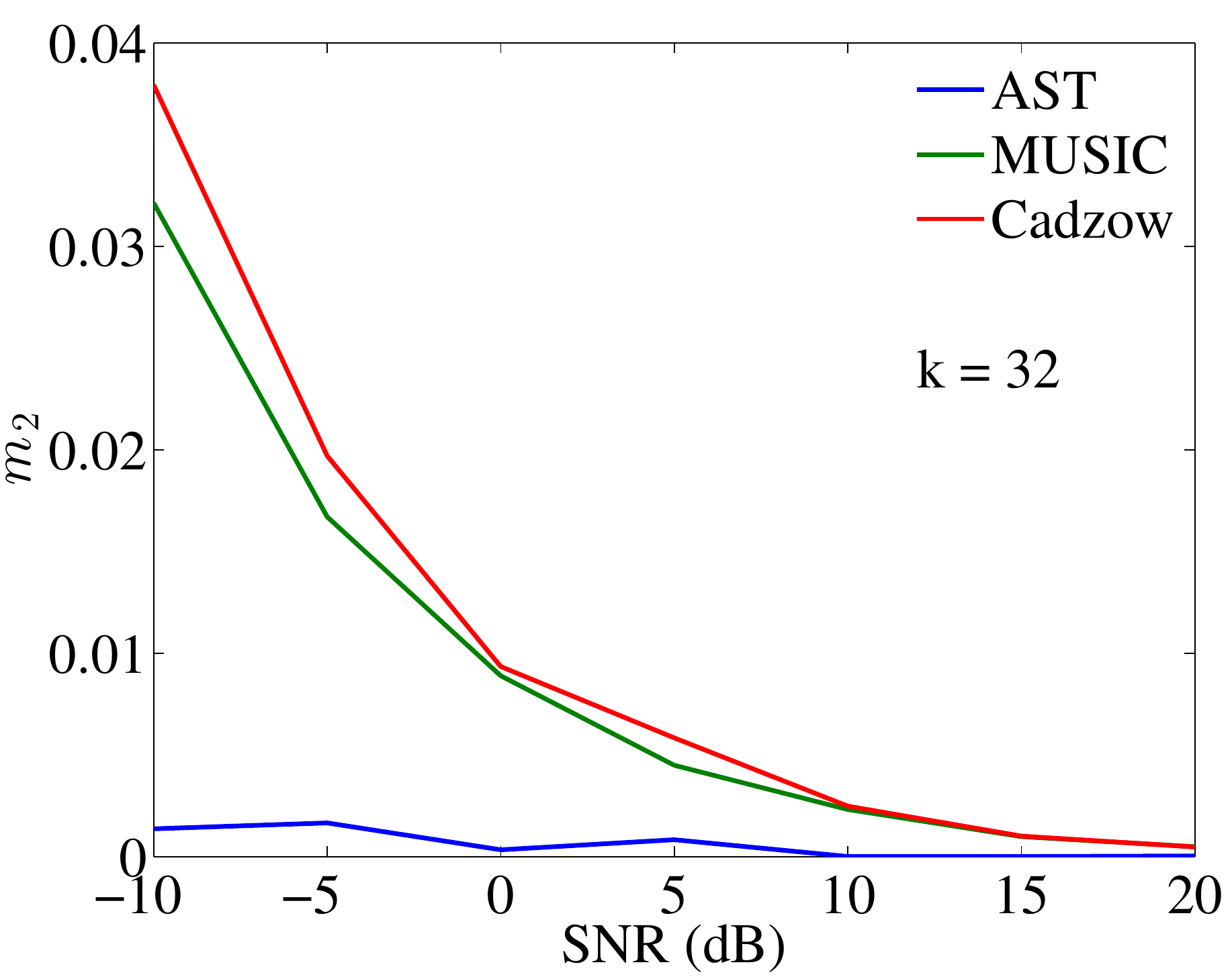} &
\includegraphics[height=43mm]{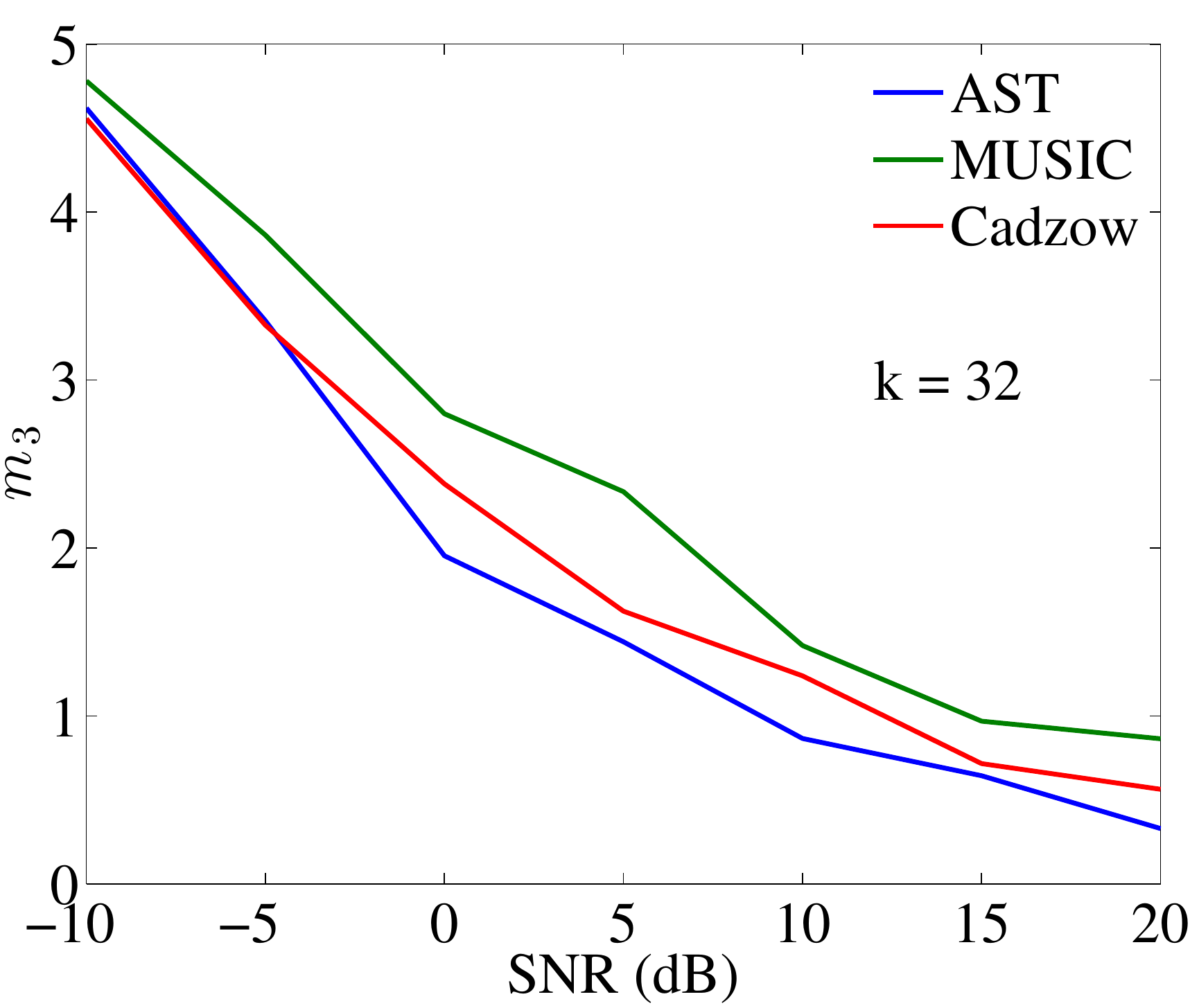} \\
\includegraphics[height=43mm]{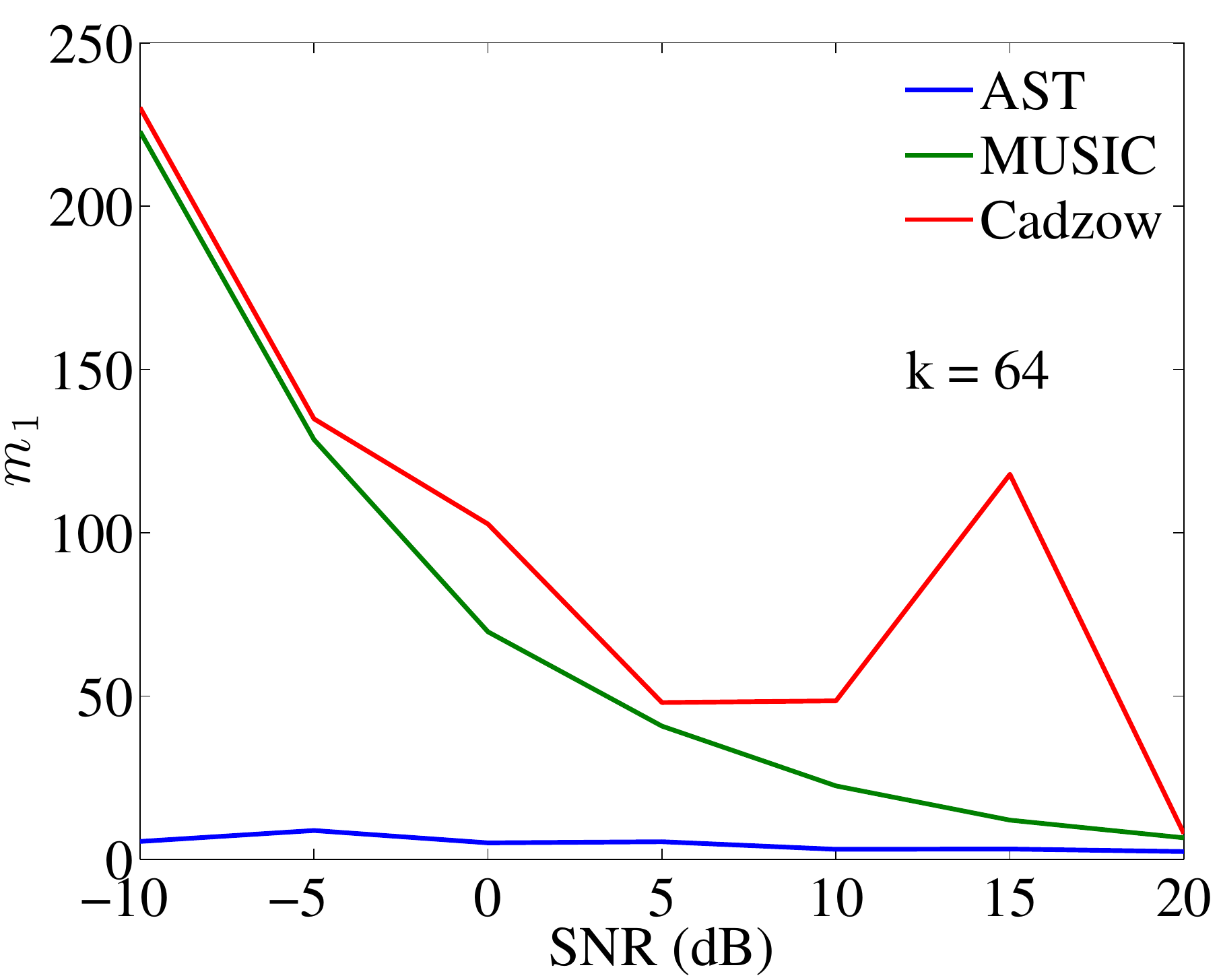} &
\includegraphics[height=43mm]{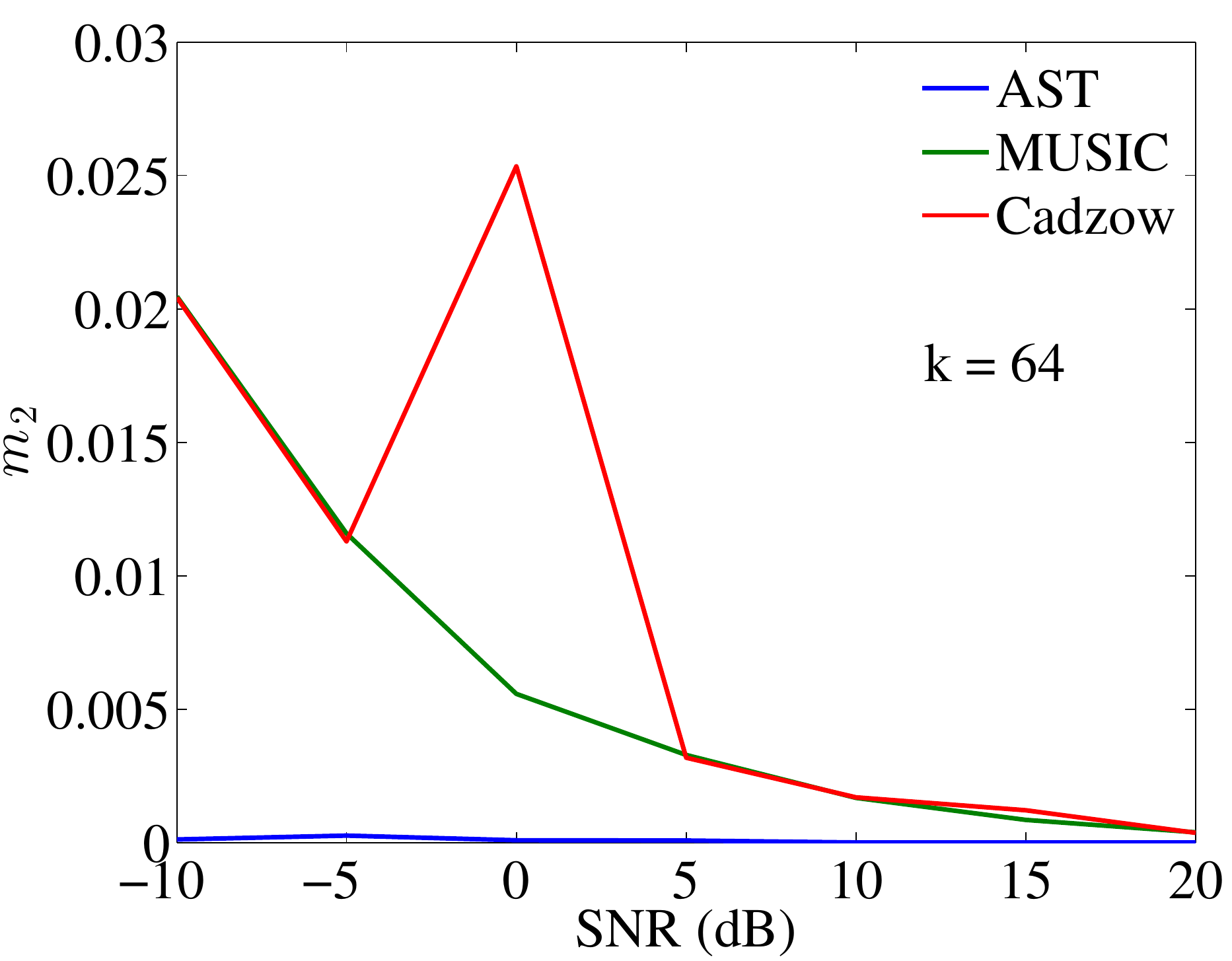} &
\includegraphics[height=43mm]{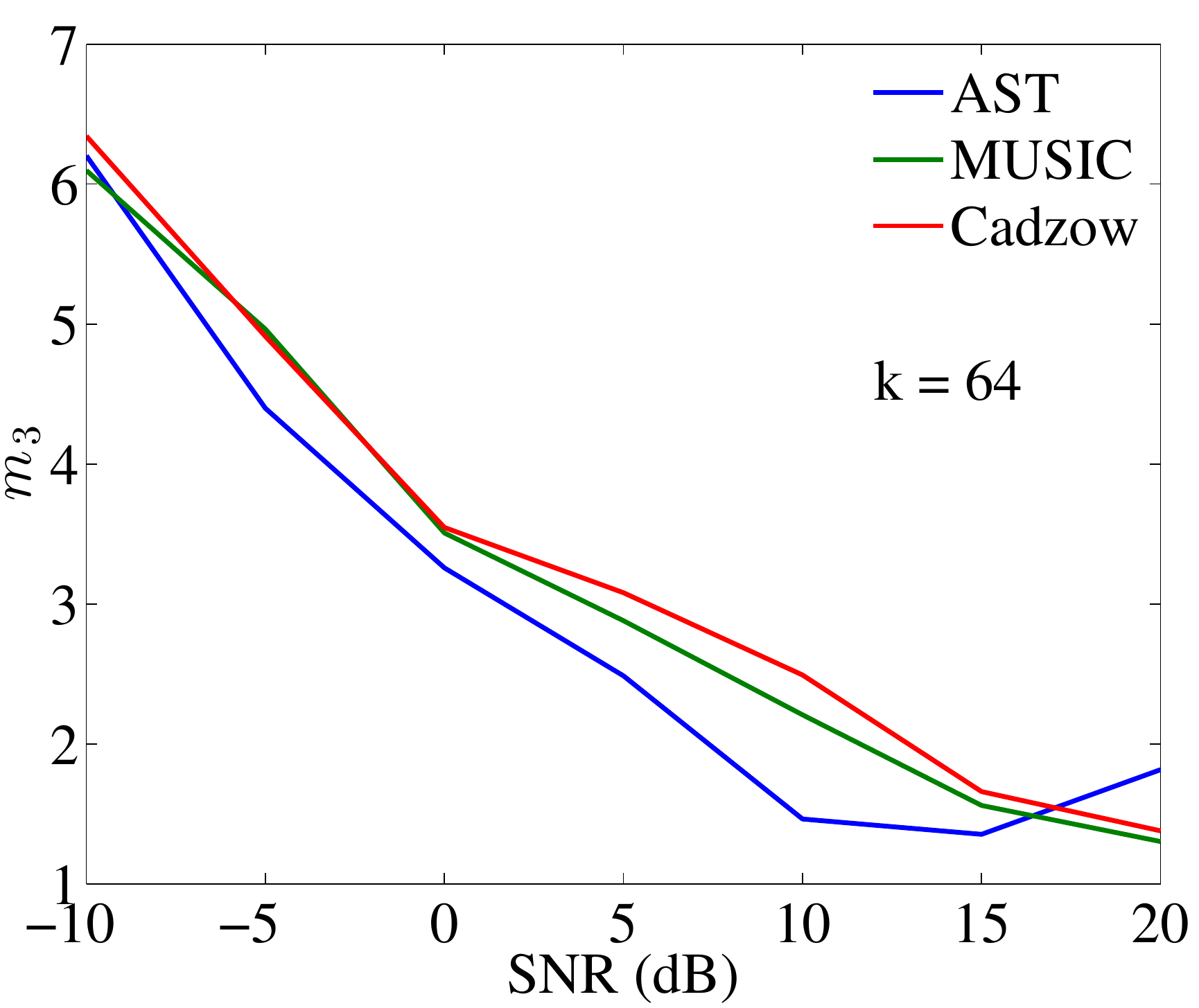}
\end{tabular}
\caption{For $n = 256$ samples, the plots from left to right in order measure the average value over 20 random experiments for the error metrics $m_1, m_2$ and $m_3$ respectively. The top, middle and the bottom third of the plots respectively represent the subset of the experiments
with the number of frequencies $k=16, 32$ and $64$.}
\label{fig:msnr}
\end{figure}

\section{Conclusion and Future Work}

In this paper, we demonstrated stability of atomic norm regularization by
analysis of specific properties of the atomic set of moments and the associated
dual space of trigonometric polynomials.  The key to our analysis is the
existence and properties of various trigonometric polynomials associated with
signals with well separated frequencies.

Though we have made significant progress at understanding the theoretical
limits of line-spectral estimation and superresolution, our bounds could still
be improved. For instance, it remains open as to whether the logarithmic term
in Theorem~\ref{main} can be improved to $\log(n/k)$. Deriving such an upper
bound or improving our minimax lower bound would provide an interesting
direction for future work.

Additionally, it is not clear if our localization bounds in
Theorem~\ref{support} have the optimal dependence on the number of sinusoids
$k$. For instance, we expect that the condition on signal amplitudes for
approximate support recovery should not depend on $k$, by comparison with
similar guarantees that have been established for Lasso~\cite{nearideal}. We
additionally conjecture that for a large enough regularization parameter, there
will be no spurious recovered frequencies in the solution. That is, there
should be no non-zero coefficients in the ``far region'' $F$ in
Theorem~\ref{support}. Future work should investigate whether better guarantees
on frequency localization are possible.

\bibliographystyle{ieeetr}
\bibliography{optand}
\appendix
\section{Proof of Lemma \ref{part1}}\label{apx:pf:taylor}
We first split the domain of integration into the near and far regions.
\begin{align}
\left |\int_0^1 X(f) \nu(df)\right | 
&\leq \left |\int_F X(f) \nu (f)\right | + \sum_{j=1}^k \left | \int_{N_j}X(f) \nu(df)\right |\nonumber \\
&\leq \vnorm{X(f)}_\infty \int_F |\nu| (df) + \sum_{j=1}^k \left | \int_{N_j}X(f) \nu(df) \right |.\label{Xfbd}
\end{align}
by using H\"{o}lder's inequality for the last inequality. Using Taylor's theorem, we may expand the integrand $X(f)$ around $f_j$ as
\[
X(f) = X(f_j) + (f-f_j) X'(f_j) + \frac{1}{2} X''(\xi_j) (f-f_j)^2 
\]
for some $\xi_j \in N_j$. 
Thus,
{\small
\begin{align*}
&|X(f)-X(f_j)-X'(f_j)(f-f_j)|\\
&\leq \sup_{\xi \in N_j} \frac{1}{2}|X''(\xi)|(f -f_j)^2\\ &\leq \frac{1}{2} n^2 \vnorm{X(f)}_\infty(f - f_j)^2, 
\end{align*}
}
where for the last inequality we have used a theorem of Bernstein for trigonometric polynomials (see, for example~\cite{bernstein}):  
\begin{align*}
|X'(f_j)|  & \leq n \vnorm{X(f)}_\infty\\
|X''(f_j)| & \leq n^2 \vnorm{X(f)}_\infty.
\end{align*}
As a consequence, we have
\begin{align*}
\left | \int_{N_j} X(f) \nu(df)\right| &\leq \left| X(f_j)\right| \left| \int_{N_j} \nu (df)\right| + \left|X'(f_j)\right| \left|\int_{N_j} (f-f_j) \nu (df)  \right|\\
& + \frac{1}{2} n^2 \|X(f)\|_\infty \int_{N_j} (f-f_j)^2 |\nu| (df) \\
& \leq \|X(f)\|_\infty \left(I_0^j + I_1^j + I_2^j\right).
\end{align*}
Substituting back into \eqref{Xfbd} yields the desired result.

\section{Some useful lemmas}
\label{apx:collection}
In addition to Theorem~\ref{dual-stab}, we recall another result in~\cite{cg_noisy} where the authors show the existence of a 
trigonometric polynomial $Q_1$ that is linear in each $N_j$ which is also an essential ingredient in our proof.

\begin{theorem}[Lemma 2.7 in~\cite{cg_noisy}]
\label{dual-lin}
For any $f_1, \ldots, f_k$ satisfying \eqref{min-sep} and any sign vector $v \in \C^k$ with $|v_j|=1$, there exists a polynomial $Q_1 = \left<q_1, a(f)\right>$ for some $q_1 \in \C^n$ with the following properties: 
\begin{enumerate}
\item For every $f \in N_j,$ there exists a numerical constant $C_a^1$ such that
\begin{equation}
\label{ca1}
|Q_1(f) - v_j(f-f_j)| \leq \frac{n}{2} C_a^1 (f-f_j)^2
\end{equation}
\item For $f \in F$, there exists a numerical constant $C_b^1$ such that
\begin{equation}
\label{cb1}
|Q_1(f)| \leq \frac{C_b^1}{n}.
\end{equation}
\end{enumerate}
\end{theorem}

We will also need the following straightforward consequence of the constructions of the polynomials in Theorem \ref{dual-stab}, Theorem \ref{dual-lin}, and Section \ref{sec:support}.
\begin{lemma}
\label{l1}
There exists a numerical constant $C$ such that the constructed $Q(f)$ in Theorem \ref{dual-stab}, $Q_1(f)$ in Theorem \ref{dual-lin}, and $Q_j^\star(f)$ in Section \ref{sec:support} satisfy respectively
\begin{align}
\|Q(f)\|_1 &:= \int_0^1{| Q(f)| df} \leq \frac{C k}{n}\label{QL1}\\
\| Q_1(f)\|_1 &\leq \frac{C k}{n^2}\label{Q1L1}\\
\|Q_j^\star\|_1 & \leq \frac{Ck}{n}\label{QjL1}.
\end{align}
\end{lemma}
\begin{proof}
We will give a detailed proof of \eqref{QL1}, and list the necessary modifications for proving \eqref{Q1L1} and \eqref{QjL1}. The dual polynomial $Q(f)$ constructed in \cite{cg_exact12} is of the
form
\begin{eqnarray}
  Q \left( f \right) & = & \sum_{f_j \in T} \alpha_j K \left( f - f_j \right)
  + \sum_{f_j \in T} \beta_j K' \left( f - f_j \right) \label{formofQ}
\end{eqnarray}
where $K \left( f \right)$ is the squared Fej\'er kernel (recall that $m = (n-1)/2$)
\begin{eqnarray*}
  K \left( f \right) & = & \left( \frac{\sin \left( \left( \frac{m}{2} + 1
  \right) \pi f \right)}{\left( \frac{m}{2} + 1 \right) \sin \left( \pi f
  \right)} \right)^4
\end{eqnarray*}
 and for $n \geq 257$, the coefficients
$\alpha \in \C^k$ and $\beta \in \C^k$ satisfy \cite[Lemma 2.2]{cg_exact12}
\begin{eqnarray*}
  \left\| \alpha \right\|_{\infty} & \leq & C_\alpha\\
  \left\| \beta \right\|_{\infty} & \leq & \frac{C_\beta}{n} 
\end{eqnarray*}
for some numerical constants $C_\alpha$ and $C_\beta$. 
Using \eqref{formofQ} and triangle inequality, we bound $\|Q(f)\|_1$ as follows: 
\begin{eqnarray}
  \|Q(f)\|_1 &=& \int_0^1 \left| Q \left( f \right) \right| d f\nonumber\\
  &\leq & k \left\| \alpha \right\|_\infty \int_0^1 \left| K \left( f\right) \right| d f + k \left\| \beta \right\|_\infty \int_0^1 \left| K' \left( f\right) \right| d f\label{Qbdgeneral1}\\
  & \leq & C_\alpha k \int_0^1 \left| K \left(f\right) \right| d f + \frac{C_\beta}{n} k \int_0^1 \left|K'(f)\right|d f \label{Q1bd},
 \end{eqnarray}
 
To continue, note that $\int_0^1 | K ( f) | d  f = \int_0^1 | G ( f) |^2 d
 f =: \|G(f)\|_2^2$ where $G ( f)$ is the Fej\'er kernel, since $K(f)$ is the squared
Fej\'er kernel. We can write
\begin{eqnarray}
  G ( f)  =  \left( \frac{\sin \left( \pi \left( \frac{m}{2} + 1 \right) f
  \right)}{\left( \frac{m}{2} + 1 \right) \sin ( \pi f)} \right)^2
   =  \sum_{l = - m / 2}^{m / 2} g_l e^{- i 2 \pi f l}\label{expressionG}
\end{eqnarray}
where $g_l = \left( \frac{m}{2} + 1 - | l | \right) / \left( \frac{m}{2} + 1
\right)^2$. Now, by using Parseval's identity, we obtain
\begin{eqnarray}
 \int_0^1 |K(f)| df & = & \int_0^1 | G ( f) |^2 d f
  =  \sum_{l = - m / 2}^{m / 2} | g_l |^2\nonumber\\
  & = & \frac{1}{\left( \frac{m}{2} + 1 \right)^4} \left( \left( \frac{m}{2}
  + 1 \right)^2 + 2 \sum_{l = 1}^{m / 2} \left( \frac{m}{2} + 1 - l \right)^2
  \right)\nonumber\\
  & = & \frac{1}{\left( \frac{m}{2} + 1 \right)^4} \left( \left( \frac{m}{2}
  + 1 \right)^2 + 2 \sum_{l = 1}^{m / 2} l^2 \right)\nonumber\\
  & \leq & \frac{C}{n}\label{bdKf}
\end{eqnarray}
for some numerical constant $C$ when $n = 2 m + 1 \geq 10$.

Now let us turn our attention to $\int_0^1 | K' (
f) | d  f$. Since $K ( f) = G ( f)^2$, we have
\begin{eqnarray}
  \int_0^1 | K' ( f) | d  f  =  2\int_0^1 | G ( f) G' ( f) | d
   f
   \leq  2\| G ( f) \|_2 \| G' ( f) \|_2\label{holder}
\end{eqnarray}
We have already established that $\| G ( f) \|_2^2 \leq C / {n}$ and we
will now show that $\| G' ( f) \|_2^2 \leq C' {n}$. Differentiating the
expression for $G(f)$ in \eqref{expressionG}, we get
\begin{eqnarray*}
  G' ( f) & = & -2 \pi i \sum_{l = - m / 2}^{m / 2} l g_l e^{- i 2 \pi f l}
\end{eqnarray*}
Therefore, by applying Parseval's identity again, we get
\begin{eqnarray*}
  \| G' ( f) \|_2^2 
  & = & 4 \pi^2 \sum_{l = - m / 2}^{m / 2} l^2 | g_l |^2\\
  & \leq &  \pi^2 m^2 \sum_{l = - m / 2}^{m / 2} | g_l |^2\\
  & \leq & C' n
\end{eqnarray*}
Plugging back into \eqref{holder} yields
\begin{align}
\int_0^1 |K'(f)| df \leq C \label{bdK1f}
\end{align}
for some constant $C$. Combining \eqref{bdK1f} and \eqref{bdKf} with \eqref{Q1bd} gives the desired result in \eqref{QL1}.

The dual polynomial $Q_1(f)$ is also of the form \eqref{formofQ} with coefficient vectors $\alpha_1$ and $\beta_1$, which satisfy \cite[Proof of Lemma 2.7]{cg_noisy}
\begin{align*}
\|\alpha_1\|_\infty \leq \frac{C_{\alpha_1}}{n},\\
\|\beta_1\|_\infty \leq \frac{C_{\beta_1}}{n^2}.
\end{align*}
Combining the above two bounds with \eqref{Qbdgeneral1}, \eqref{bdK1f} and \eqref{bdKf} gives the desired result in \eqref{Q1L1}.

The last polynomial $Q_j^\star$ also has the form \eqref{formofQ} with coefficient vectors $\alpha^\star$ and $\beta^\star$. According to \cite[Proof of Lemma 2.2]{granda2}, these coefficients satisfy
\begin{align*}
\|\alpha^\star\|_\infty \leq {C_{\alpha_\star}},\\
\|\beta_\star\|_\infty \leq \frac{C_{\beta_\star}}{n},
\end{align*}
which yields \eqref{QjL1} following the same argument leading to \eqref{QL1}. 

\end{proof}

Using Lemma \ref{l1}, we can derive the estimates we need in the following lemma.
\begin{lemma}
\label{l4}
Let $\nu = \hat{\mu} - \mu$ be the difference measure. Then, there exists numerical constant $C>0$ such that
\begin{align}
\label{qv}\left| \int_0^1 Q(f) \nu(df) \right| &\leq \frac{C k \tau}{n}\\
\label{q1v}\left| \int_0^1 Q_1(f) \nu(df) \right| &\leq \frac{C k \tau}{n^2}\\
\label{qjv} \left| \int_0^1 Q_j^\star(f) \nu(df) \right| & \leq \frac{Ck\tau}{n}.
\end{align}
\end{lemma}
\begin{proof}
Let $Q_0 = \langle q_0, a(f) \rangle $ be a general trigonometric polynomial associated with $q_0 \in \C^n$. Then,
\begin{align*}
\left|\int_0^1 Q_0(f) \nu(df) \right| 
& = \left|\int_0^1 \langle q_0 , a(f) \rangle  \nu(df) \right|\\
& = \left|\langle q_0,  \int_0^1  a(f)  \nu(df) \rangle\right|\\
& = \left|\langle q_0, e \rangle\right|\\
& = \left|\langle Q_0(f), E(f) \rangle\right|\\
& \leq \vnorm{Q_0(f)}_1 \vnorm{E(f)}_\infty\,.
\end{align*}
Here we use Parseval's identity in the second to last step and H\"{o}lder's inequality in the last inequality. Then, the result follows by using Lemma~\ref{l1} and \eqref{errbd}.
\end{proof}

We also need the following consequence of the optimality condition of AST from~\cite[Lemma 2]{btr12}:
\begin{proposition}\label{pro:optimality}
\begin{align}
\tau \vnorm{\hat{x}}_\A \leq \tau \vnorm{x^\star}_\A + \langle w, \hat{x} - x^\star \rangle
\end{align}
\end{proposition}

\section{Proof of Lemma \ref{part2}}\label{apx:pf:I0I1}
Consider the polar form
\[
  \int_{N_j} \nu ( df)  =  \left| \int_{N_j} \nu ( df) \right| e^{i \theta_j} .
\]
Set $v_j = e^{-i \theta_j}$ and let $Q(f)$ be the dual polynomial promised by Theorem \ref{dual-stab} for this $v$. Then, we have 
\begin{align*}
  \left| \int_{N_j} \nu ( d f) \right| & = 
  \int_{N_j} e^{- i \theta_j} \nu ( d f)\\
  & =  \int_{N_j} Q ( f) \nu ( d f) + 
  \int_{N_j} (e^{- i \theta_j} -  Q ( f) ) \nu ( d f)
\end{align*}
Summing over $j=1,\ldots,k$ yields

\begin{align}
\nonumber I_0 &= \sum_{j=1}^k \left| \int_{N_j} \nu(df) \right|\\
\nonumber & = \sum_{j=1}^k\int_{N_j} Q ( f) \nu ( d f) + 
\sum_{j=1}^k \int_{N_j} ( v_j - Q ( f)) \nu ( d f)\\
\nonumber & \leq \left|\int_0^1 Q(f) \nu(df) \right| + \int_F|\nu|(df) + C_a' I_2, \text{ using triangle inequality and \eqref{q2}}\\
\label{i0} & \leq \frac{C k \tau}{n} + \int_F|\nu|(df) + C_a' I_2, \text{ using \eqref{qv}}.
\end{align}
We use a similar argument for bounding $I_1$ but this time use the dual polynomial $Q_1(f)$ guaranteed by Theorem \ref{dual-lin}. Again, start with the polar form
\[
  \int_{N_j} (f - f_j) \nu ( d f)  =  \left|
  \int_{N_j} (f - f_j) \nu ( d f) \right| e^{i \theta_j} = I_1^j e^{i\theta_j}/n
\]
Set $v_j = e^{-i \theta_j}$ in Theorem \ref{dual-lin} to obtain
{
\begin{align*}
  I_1^j & = 
  n\int_{N_j} e^{- i \theta_j} ( f - f_j) \nu ( d f)\\
  & =  n \int_{N_j} (v_j (
  f - f_j) - Q_1 ( f)) \nu ( d f)  + n\int_{N_j} Q_1 ( f) \nu ( d f)
\end{align*}
}
Summing over $j=1,\ldots,k$ yields
{
\begin{align}
\nonumber I_1 &= \sum_{j=1}^k I_1^j\\
\nonumber &= n \sum_{j=1}^k \int_{N_j} (v_j (
  f - f_j) - Q_1 ( f)) \nu ( d f) + n\sum_{j=1}^k\int_{N_j} Q_1 ( f) \nu ( d f)\\
\nonumber &\leq C_a^1 I_2 + n\left|\int_0^1 Q_1(f) \nu(df)\right| +  n\left |\int_F Q_1(f) \nu(df)\right |\\
\label{i1}& \leq C_a^1 I_2 + \frac{C k \tau}{n} +  C_b^1 \int_F|\nu|(df)
\end{align}
}
For the first inequality, we have used \eqref{ca1} and triangle inequality, and for the last inequality, we have used \eqref{q1v} and \eqref{cb1}. Equations \eqref{i0} and \eqref{i1} complete the proof.

\section{Proof of Lemma \ref{part3}}
\label{apx:pf:I2far}
Denote by $P_T(\nu)$ the projection of the difference measure $\nu$ on the support set $T = \{f_1, \ldots, f_k\}$ of $x^\star$ so that $P_T(\nu)$ is supported on $T$. Then, setting $Q(f)$ the polynomial in Theorem \ref{dual-stab} that interpolates the sign of $P_T( \nu)$, we have
{
\begin{align*}
  \| P_T ( \nu) \|_{\mathrm{TV}} & =  \int_0^1 Q ( f) P_T ( \nu) ( d f)\\
  & \leq  \left| \int_0^1 Q ( f) \nu ( d f) \right| + \left|
  \int_{T^c} Q ( f) \nu ( d f) \right|\\
  & \leq  \frac{C k \tau}{n} + \sum_{f_j \in T} \left|
  \int_{N_j / \{ f_j \}} Q ( f) \nu ( d f) \right| + \left|
  \int_F Q ( f) \nu ( d f) \right|,
\end{align*}}
where for the first inequality we used triangle inequality and for the last inequality we used \eqref{qv}. 
The integration over $F$ is can be bounded using H\"{o}lder's inequality
\[
  \left| \int_{F} Q ( f) \nu ( d f) \right|  
  \leq  ( 1 - C_b) \int_F |\nu|(df)
\]

We continue with
{
\begin{align*}
  \left| \int_{N_j / \{ f_j \}} Q ( f) \nu ( d f) \right| 
  & \leq  \left| \int_{N_j / \{ f_j \}} | Q ( f) | | \nu | ( d f) \right|\\
  & \leq  \int_{N_j / \{ f_j \}} ( 1 - \tfrac{1}{2}n^2 C_a ( f - f_j)^2) | \nu | ( d f)\\
  & \leq \int_{N_j / \{ f_j \}} | \nu | ( d f) - C_a I_2^j.
\end{align*}
}
As a consequence, we have
{
\begin{align*}
  \nonumber \| P_T ( \nu) \|_{\mathrm{TV}} & \leq  \frac{C k \tau}{n} + \sum_{f_j
  \in T} \int_{N_j / \{ f_j \}} | \nu | ( d f) - C_a
  I_2 + ( 1 - C_b) \int_F |\nu|(df) \nonumber\\
 & \leq  \frac{C k \tau}{n} + \underbrace{\sum_{f_j
  \in T} \int_{N_j / \{ f_j \}} | \nu | ( d f) + \int_F |\nu|(df)}_{\|P_{T^c}\|_{\mathrm{TV}}} - C_a
 I_2  - C_b \int_F |\nu|(df)
\end{align*}
}
or equivalently,
\begin{align}
  \label{eqn:lower}
\|P_{T^c}(\nu) \|_{\mathrm{TV}} - \|P_T(\nu)\|_{\mathrm{TV}} \geq C_a I_2 + C_b \int_F |\nu|(df) - \frac{C k\tau}{n}.
\end{align}

Now, we appeal to Proposition \ref{pro:optimality} and obtain
\[
\vnorm{\hat{x}}_\A \leq \vnorm{x^\star}_\A - \langle w, e \rangle/\tau
\]
and thus
\begin{equation}
\label{opt-cond}
\vnorm{\hat{\mu}}_{\mathrm{TV}} \leq \vnorm{\mu}_{\mathrm{TV}} + |\langle w, e \rangle|/\tau.
\end{equation}
Using Lemma~\ref{part1},
\begin{align}
\nonumber |\langle w, e \rangle| & = |\langle w, \int_0^1 a(f) \nu(df)  |\rangle\\
& = \left|\int_0^1  \left\langle w,  a(f)  \right\rangle \nu(df)\right|\\
\nonumber & \leq \vnorm{\left\langle w,  a(f)  \right\rangle}_\infty\left(\frac{C k \tau}{n} + I_0 + I_1 + I_2\right)\\
\label{w-expand}& \leq 2\eta^{-1} \tau \left(\frac{C k \tau}{n} + I_0 + I_1 + I_2\right)\nonumber \\
 & \leq C \eta^{-1} \tau \left(\frac{k\tau}{n} + I_2 + \int_F |\nu|(df) \right)
\end{align}
with high probability, where for the penultimate inequality we used our choice of $\tau$ and $\vnorm{\left\langle w,  a(f)  \right\rangle}_\infty \leq 2\eta^{-1} \tau$ with high probability, a fact shown in Appendix C of ~\cite{btr12}. 
Substituting \eqref{w-expand} in \eqref{opt-cond}, we get
\begin{align*}
& \vnorm{\mu}_{\mathrm{TV}} + C \eta^{-1} \tau \left(\frac{k\tau}{n} + I_2 + \int_F |\nu|(df) \right)\\
& \geq \vnorm{\hat{\mu}}_{\mathrm{TV}}\\
& = \vnorm{\mu + \nu}_{\mathrm{TV}}\\
& \geq \vnorm{\mu}_{\mathrm{TV}} - \vnorm{{P_T(\nu)}}_{\mathrm{TV}} + \vnorm{{P_{T^c}(\nu)}}_{\mathrm{TV}}\end{align*}
Canceling $\|\mu\|_{\mathrm{TV}}$ yields 
\begin{align}\label{eqn:upper}
\|P_{T^c}(\nu)\|_{\mathrm{TV}} - \|P_T(\nu)\|_{\mathrm{TV}} \leq C\eta^{-1}\tau \left(\frac{k \tau}{n} + I_2 + \int_F |\nu|(df)\right)
\end{align}
As a consequence of \eqref{eqn:lower} and \eqref{eqn:upper}, we get,
\[
  C(1+\eta^{-1}) \frac{k \tau}{n} \geq  ( C_b - \eta^{-1} C)  \int_F{|\nu|(df)} + ( C_a - \eta^{-1} C)I_2
\]
whence the result follows for large enough $\eta.$
\end{document}